\documentclass[conference]{IEEEtran}
\IEEEoverridecommandlockouts

\usepackage{graphicx}  
\usepackage{graphbox}  
\usepackage{mathrsfs}
\usepackage{color}
\usepackage{amsfonts,amssymb}
\usepackage{amsfonts}
\usepackage{subfigure}
\usepackage{booktabs}

\usepackage{amsmath}
\usepackage{enumerate}
\usepackage[ruled,vlined,algo2e]{algorithm2e}
\usepackage{bm}
\usepackage{makecell}

\usepackage{bbm}
\usepackage{dsfont}
\usepackage{pifont}
\usepackage{amsthm}
\usepackage{cite}

\usepackage{enumitem}
\usepackage{balance}

\newtheorem{remark}{Remark}
\newtheorem{theorem}{Theorem}
\newtheorem{lemma}{Lemma}
\newtheorem{corollary}{Corollary}
\newtheorem{assumption}{Assumption}
\newtheorem{problem}{Problem}

\begin{document}

\title{Observation-based Optimal Control Law Learning with LQR Reconstruction}
\author{Chendi Qu, Jianping He, Xiaoming Duan 
	\thanks{
	 The authors are with the Dept. of Automation, Shanghai Jiao Tong University, and Key Laboratory of System Control and Information Processing, Ministry of Education of China, Shanghai, China. E-mail address: \{qucd21, jphe, xduan\}@sjtu.edu.cn. 
	}
 \thanks{Part of this paper has been accepted to IFAC 2023 World Congress \cite{qu2023control}.}
}

\maketitle

\begin{abstract}
Designing controllers to generate various trajectories has been studied for years, while recently, recovering an optimal controller from trajectories receives increasing attention. In this paper, we reveal that the inherent linear quadratic regulator (LQR) problem of a moving agent can be reconstructed based on its trajectory observations only, which enables one to learn the optimal control law of the agent autonomously. Specifically, the reconstruction of the optimization problem requires estimation of three unknown parameters including the target state, weighting matrices in the objective function and the control horizon. 
Our algorithm considers two types of objective function settings and identifies the weighting matrices with proposed novel inverse optimal control methods, providing the well-posedness and identifiability proof. We obtain the optimal estimate of the control horizon using binary search and finally reconstruct the LQR problem with above estimates. 
The strength of learning control law with optimization problem recovery lies in less computation consumption and strong generalization ability. We apply our algorithm to the future control input prediction and the discrepancy loss is further derived.
Numerical simulations and hardware experiments on a self-designed robot platform illustrate the effectiveness of our work. 
\end{abstract}

\section{Introduction}
Nowadays, as the development of localization, computer vision and planning, mobile agents have been widely applied in various fields and achieved high-profile success \cite{cao2012mobile}. Lots of studies focus on designing controllers to generate satisfying trajectories \cite{tzafestas2018mobile}, while in this paper we pay attention to its inverse problem: \emph{How to learn the interior controller based on agent's trajectory observations?}

To tackle this problem, let's consider a mobile agent driving from its initial position to a target point, which is one of the most common and basic scenarios of agents. In most situations, the agent is exposed to the physical world thus its trajectory can be observed by potential external attackers \cite{zhou2023robust,pasqualetti2013attack, wu2023secure}. We suppose there is another agent equipped with camera and computation ability trying to learn the trajectory and motion of the agent based on observations only.
The issue that agents learn from other agents through observations is similar to learning from demonstration (LfD) \cite{ravichandar2020recent,jin2022learning}, while in our case the agent needs to collect data actively instead of just being provided by human demonstrations. LfD has raised extensive studies recent years and applied to various fields including autonomous driving \cite{kuderer2015learning}, manufacturing \cite{kent2016construction} and human-robot interaction \cite{maeda2017probabilistic}. Given standard demonstrations, one mainstream category of LfD algorithms is to learn the control policy directly from the state observations to action \cite{rahmatizadeh2018vision,torabi2018behavioral}, known as end-to-end learning. However, this branch of methods usually require for large amount of demonstrations as training data and contain little generalization ability between different environments.


Therefore, our approach is to conduct a two-stage LfD, learning the interior controller of the agent first, which is a non-trivial problem, since the control objective of the agent is unknown. To further simplify the problem, we assume the moving agent is utilizing an LQR controller \cite{bemporad2002explicit}. Since LQR is a state-feedback controller, the control law is actually a series of feedback gain matrices. Note that if the controller is infinite-time, these matrices are invariant and can be estimate with the optimal state and input trajectories. However, in the finite-time case, this sequence of matrices is time invariant with unknown length.
Thus, in our paper the main idea is to reconstruct the control optimization problem. Once the optimization problem is recovered, we can obtain the control law and imitate the agent's motion with strong generalization performance. Assuming the dynamic model and objective function form is known, to reconstruct the LQR problem, we need to obtain the target state, the weighting matrices and the control horizon.

An inspiring two-stage LfD method is
inverse optimal control (IOC) \cite{ab2020inverse, jin2019inverse}, which is utilized to identify objective function parameters being an important part of our algorithm. IOC algorithms diverse according to the assumption of the objective function form. As for the quadratic LQR form considered in our case, \cite{anderson2007optimal, priess2014solutions} both study the infinite-time problem assuming the constant control gain matrix is already known, while we provide the gain matrix estimator. 
In addition, \cite{pauwels2016linear, li2020continuous} propose approaches to identify objective function as continuous finite-time LQR and \cite{yu2021system,zhang2019inverse} solve IOC problems for discrete-time. However, few of these studies consider the classic LQR setting including final-state, process-state and process-cost terms with different parameters, and they usually require complete trajectory observations.
The past research on control horizon has mainly focused on adaptive horizon model prediction control. 
The control stability and effectiveness of the system can be decided by the horizon length \cite{primbs2000feasibility} and the horizon estimation refers to choosing an optimal length to balance the performance and computation cost \cite{bohn2021reinforcement,sun2019robust}. However in our paper, we study how to identify the horizon length of the optimization problem given optimal state trajectory observations with noises.

One important application of our optimization problem reconstruction algorithm is the control input and trajectory prediction, which can be the basis for attacker's subsequent attacks, interception or misleading. Some existing methods for trajectory prediction are data-driven and model-free, such as using polynomial regression \cite{chen2016tracking,gao2018online} or introducing neural networks, including the long short-term memory neural network \cite{altche2017lstm} or graph neural network \cite{mohamed2020social}. However, these usually require a large amount of observations as well as the model training in the early stage. Another class of algorithms is model-based.
For instance, \cite{schulz2018multiple} use an unscented Kalman filter to predict multi-agent trajectories. \cite{li2020unpredictable} measure the secrecy of the trajectory and proves that uniform distributed inputs maximize the unpredictability of the system. 
But these predictive models often assume that the control inputs at each step are known, which in our case, cannot be obtained directly.
Notice that if we build the reconstructed optimization problem of the agent, we are able to calculate the input at arbitrary states and predict the future trajectory accurately, combining with the estimation of the dynamic model and current state through system identification (SI) \cite{ljung1998system} and data fusion filters \cite{welch1995introduction,gillijns2007unbiased}, respectively.

Motivated by above discussion, we design a control law learning algorithm with LQR problem reconstruction based on trajectory observations. 
We first estimate the target point by finding the intersection of the trajectory extension lines. Then identify the weighting matrices in the objective function with IOC algorithms and obtain the optimal horizon estimate through a binary search method. Finally, we reconstruct the control optimization problem and solve its solution as the learned control law.
This algorithm can not only learn the control law and enable agents to imitate the motion from arbitrary initial states, but to predict the future input and state trajectory of the mobile agent precisely. One of the main challenges of problem reconstruction is how to build a proper IOC problem to identify the parameters simultaneously in complex quadratic objective function forms. Another lies in analyzing the estimation error of each part and their impact on the input prediction application.

This paper is an extension of our conference paper \cite{qu2023control}. The main differences include i) we reorganize the article structure to focus on LQR problem reconstruction, treating input prediction as an application, ii) the infinite-time case is considered, iii) the IOC problem for more complex form of the objective function is built and solved, iv) the sensitivity analysis of control horizon estimation is provided, v) extended simulations are provided and hardware experiments are added.
The main contributions are summarized as follows:
\begin{itemize}
\item We investigate the observation-based optimal control law learning issue and propose a novel LQR reconstruction algorithm for mobile agents, including the estimation of target state, objective function parameters and the control horizon. As far as we know, we are the first to consider and reveal that the entire optimization LQR problem can be recovered based on only observations.
\item We solve the IOC problem considering two forms of objective functions, including i) final-state only setting based on PMP conditions given incomplete state trajectories and ii) classic LQR setting based on condition number minimization. 
We provide the scalar ambiguity property analysis and further prove the uniqueness and identifiability of the problem.
Furthermore, a novel approach for estimating the agent's control horizon is presented, converting a non-convex integer optimization into a binary search problem.
\item We apply our LQR reconstruction algorithm to the future control inputs prediction problem, providing error sensitivity analysis based on the convergent property of the algebraic Riccati equation. 
\item Numerical simulations reveals the effectiveness of both objective parameters and control horizon estimations. Our algorithm shows low bias and variance of the error. Moreover, hardware experiments of input prediction conducted on our self-designed robot platform demonstrate the prediction accuracy and efficiency.
\end{itemize}

The remainder of the paper is organized as follows. Section \ref{preliminary} describes the problem of interest. Section \ref{ioc} analyzes the infinite-time case and studies the objective function identification in two settings. Section \ref{horizon} estimates the control horizon and summarizes the complete algorithm flow. Simulation results and hardware experiments are shown in Section \ref{sim}-\ref{exp}, followed by the conclusion in Section \ref{conc}.

\begin{figure}[t]
\centering
\includegraphics[width=0.38\textwidth]{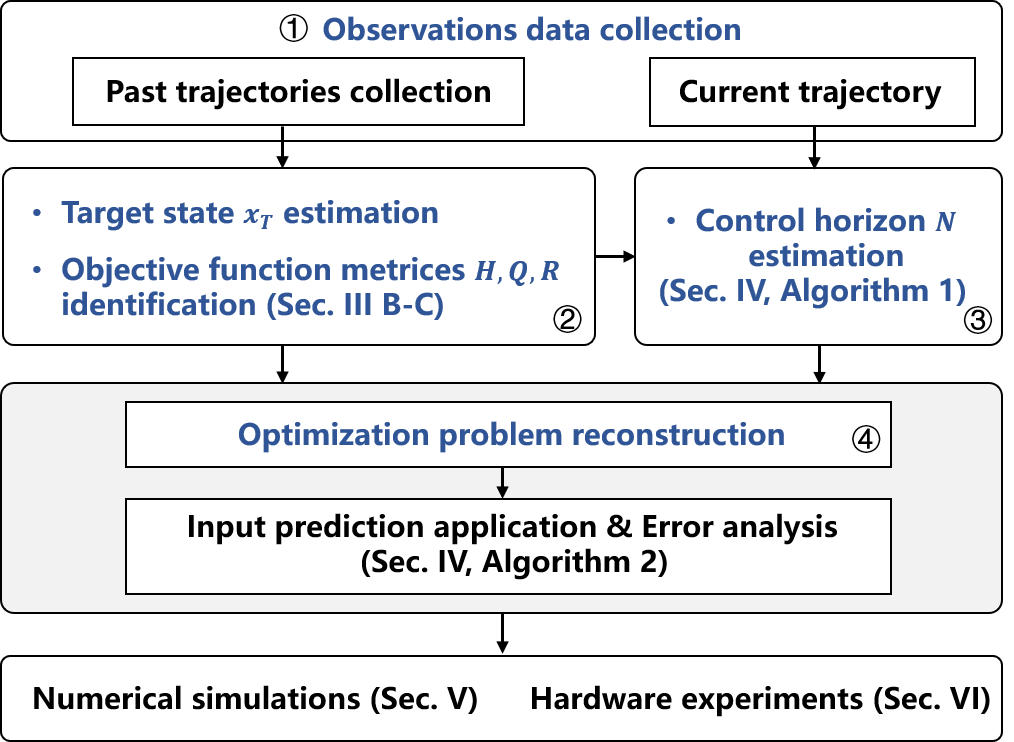}
\caption{Road-map of this paper.}
\label{alg_map}
\end{figure}

\begin{table}[t]
\small
\caption{\label{tab:test}Notation Definitions}
 \begin{tabular}{cl}
 \toprule 
  Symbol  & Definition  \\ 
  \midrule
  $\hat{x}_{k|k}$ & the estimation of state $x_k$ at time $k$;\\
  $\hat{u}_{k|k}$ & the prediction of input $u_k$ at time $k$;\\
  $x_T$ & the set target state;\\
  $A^c (A^c_k)$ & the closed-loop system matrix (at time $k$);\\
  $K (K_k)$ & the feedback gain matrix of LQR problem (at time $k$);\\
  $K_{i:j}$ & a matrix sequence $\{K_i, K_{i+1}, \cdots,K_j\}$;\\
  $N$ & the control horizon of a control process;\\
  $M$ & the number of trajectories in the observation data;\\
  $l (l_j) $ & the length of a single ($j$-th) observation trajectory;\\
  $\mathcal{Y}^j$ & the $j$-th observation trajectory;\\
  $\mathcal{X}^j$ & the state sequence estimated from observation $\mathcal{Y}^j$;\\
  $A(i,j)$ & The $i$-th row and $j$-th column element in matrix A;\\
  $A(:,j)$ & The $j$-th column of matrix A;\\
  $A^{\dagger}$ & Pseudo-inverse of matrix A;\\
  $I_n$ & an n $\times$ n dimensional identity matrix;\\
  $\|\cdot\| $ & the 2-norm of a vector; \\
  \bottomrule 
 \end{tabular} \label{table-n0}
\end{table}

\section{Preliminaries and Problem Formulation}\label{preliminary}
\subsection{Model Description}
Consider a mobile agent $R_m$ driving from its initial state to a fixed target point $x_T$. $R_m$ is modeled by a discrete-time linear system
\begin{equation}\label{sys}
{x}_{k+1} = A {x}_k + B {u}_k,
\end{equation}
where ${x}_k$ is the state vector, ${u}_k$ is the control input and $A, B$ are $n \times n$ and $ n \times m$ matrices.
Assume that (A, B) is controllable and B has
full column rank. Moreover, $A$ is an invertible matrix since the system matrix of a discrete-time system sampled from a continuous linear system is always invertible \cite{zhang2019inverse}. 
The output function is
\begin{equation}\label{ob_fun}
{y}_{k} = C {x}_k + \omega_k,
\end{equation}
where ${y}_k$ is the outputs such as agent's position and $C \in \mathbb{R}^{p \times n}$ is the observation capability matrix. We require $C$ is invertible in our algorithm and analysis since if $C$ is a fat matrix ($p <n$), the information contained in $x_t$ cannot be fully characterized by $y_t$ and the identification will be difficult \cite{molloy2016discrete}. 
$\omega_k \sim \mathcal{N}(0,\Gamma)$ is independent Gaussian observation noise. There is
\begin{equation}
\mathbb{E}(\omega_k) = 0, \, \mathbb{E}(\Vert \omega_k \Vert^2) < + \infty.
\end{equation}

Note that $R_m$ follows the optimal LQR control and the optimization problem is described as
\begin{subequations}
\begin{alignat}{2}
\mathbf{P}_0: & \min_{u_{0:N-1}} && J_0 = \frac{1}{2} x_N^T H x_N+\sum_{k = 0}^{N-1}\frac{1}{2} (x_k^T Q x_k +u_k^T R u_k), \nonumber \\ 
& ~~\mathrm{s.t.}
&& \eqref{sys}, x_0 = \bar{x}- x_T, \nonumber
\end{alignat}
\end{subequations}
where $H, R$ are positive definite matrices, $Q$ is a semi-positive definite matrix and $\bar{x}$ is the initial state. The first term in objective function $J_0$ reflects the deviation to the target state and the second term represents the cost of energy during the process. 
To simplify the problem description, we assume the agent sets the target as $0$ in $\mathbf{P}_0$ and performs coordinate system transformation on the initial state $\bar{x}$ with $x_T$.
It is known that the solution of $\mathbf{P}_0$ is
\begin{equation}\label{eq:uk}
u_k = -K_k x_k,\,k = 0,1,\dots,N-1,
\end{equation}
where $K_{0:N-1}=\{K_0,K_1,\dots,K_{N-1}\}$ is the control gain matrix sequence related to the system equation and control objectives, which is calculated through the iterative equations \eqref{dis_k} shown in Section \ref{lqr_setting}.

\subsection{Problem Formulation}
Now there is another external agent $R_o$ observing and recording the output trajectories of $R_m$. We assume $R_o$ has exact knowledge of the dynamic function, which can be accessed by SI methods, and the quadratic form of the objective function.
Suppose that $R_o$ obtains $M\geqslant n$ optimal trajectories $\{\mathcal{Y}^1, \cdots, \mathcal{Y}^M\}$ of $R_m$, where 
\[
\mathcal{Y}^j = \{y_0^j, y_1^j,\dots,y_{l_j}^j\}, j = 1, 2, \dots, M,
\]
is the $j$-th trajectory and $l_j+1 \geqslant 2n$ is the length.

\begin{assumption}\label{data_ass}
We require there exist at least $n$ linearly independent final states among $M$ trajectories, which means the matrix $\begin{bmatrix}
y_{l_1}^1 & \dots & y_{l_M}^M
\end{bmatrix}$ has a full row rank.
\end{assumption}



Consider that $R_o$ tries to learn the control law of $R_m$ based on observations, which is equal to estimate the control feedback gain matrix sequence $K_{0:N-1}$ of $\mathbf{P}_0$ accurately.
However, this is a non-trivial problem since $K_k$ is time-variant and the control objective of $R_m$ is unknown. Therefore, the main idea of this paper is to reconstruct the optimization problem $\mathbf{P}_0$ in order to solve $K_{0:N-1}$ directly, which means we need to estimate the following unknown parameters: 
\begin{itemize}
\item  target state ${x}_T$;
\item  weighting matrices in objective function $H, Q, R$;
\item  control horizon length ${N}$.
\end{itemize}

If the above estimations are accurate, we can obtain the real control law of $R_m$ accurately by solving the reconstructed optimization problem. Notice that the target state estimation is quite trivial and less important in our problem reconstruction. We have the following remark.
\begin{remark}
To estimate the target state, we can do curve fitting on at least two non-parallel observed trajectories, whose intersection is calculated as the estimate $\hat{x}_T$.
Moreover, one simple way is to directly observe the final state in multiple trajectories and take the average as the target value.
\end{remark}
Hence, in the following sections, we will estimate the objective function parameters and control horizon accordingly, and then formulate the optimization problem with all the estimates.

\section{Objective Function Identification}\label{ioc}
In this section, we will firstly discuss the control law learning in the infinite horizon case as an inspiration. Then, we focus on finite-time problem and estimate the weighting parameters in the objective function through IOC methods. 
Since in some situations, the agent only pays attention to whether the target state can be reached and the whole energy consumption during the process instead of the transient states. Therefore, we will first consider the final-state only setting based on looser data assumptions. In the third subsection, we will solve the classic LQR setting IOC problem.

\subsection{Infinite Horizon Case}\label{inf_sec}
When the control horizon $N$ of agent $R_m$ goes to infinity, the objective function in $\mathbf{P}_0$ changes into
\[
J^{\infty} = \frac{1}{2}\sum_{k = 0}^{\infty} (x_k^T Q x_k + u_k^T R u_k).
\]
and the optimal control law $K$ is given by
\begin{equation}\label{k_inf}
K =  (R + B^T P B)^{-1} B^T P A,
\end{equation}
where the intermediate parameter $P$ satisfies the following Riccati equation:
\begin{equation}\nonumber
P = A^TP A - A^T P B(R+ B^T P B)^{-1} B^TP A + Q.
\end{equation}
Since in the infinite case, the feedback gain matrix $K$ is constant, we can estimate $K$ based on observations directly without the control problem reconstruction.

Denote the closed-loop system matrix $A^c = A - BK$ and its spectral radius satisfies $\rho(A^c) < 1$. We have
\begin{equation}\nonumber
y_{k+1} = C A^c C^{-1} (y_k - \omega_k) + \omega_{k+1}
\end{equation}
for all $k$, which is a typical form of first-order stationary vector auto-regression process.
We can calculate $A^c$ inspired by ordinary least square (OLS) method with one single trajectory $\mathcal{Y}^1$.
The estimator is designed as
\begin{equation}
\hat{A}^c =
C^{-1} (\frac{1}{l} Y X^T) (\frac{1}{l} X X^T \! -\! I_l \otimes \Gamma)^{-1} C, 
\end{equation}
where $X = (y^1_0 , y^1_1, \dots, y^1_{l-1})^T$, $Y = (y^1_1 , y^1_2 ,\dots, y^1_{l})^T$ and $\Gamma = \mathrm{diag}\{\sigma^2_{\omega,1}, \cdots, \sigma^2_{\omega,p}\}$. The matrix $X^TX$ is required to be full rank to guarantee the uniqueness. 

Then, we learn the constant gain matrix
\begin{equation}
\hat{K} = B^{\dagger} (A - \hat{A}^c).
\end{equation}
as control law.
The error between the estimation and the true $K$ can be restricted by:
\begin{equation}\nonumber
\begin{aligned}
&\Vert K - \hat{K} \Vert_F  =  \Vert B^{\dagger} (A - A^c) - B^{\dagger} (A - \hat{A}^c) \Vert_F \\
 & \leqslant \Vert B^{\dagger} \Vert_F \Vert \hat{A}^c - A^c \Vert_F \leqslant \Vert B^{\dagger} \Vert_F \cdot \sqrt{n} \Vert \hat{A}^c - A^c \Vert.
\end{aligned}
\end{equation}
According to Theorem 6 in \cite{li2023topology}, the estimation error of gain matrix $K$ converges to zero as the number of samples $l$ goes larger, which is
\begin{equation}
\lim_{l \rightarrow \infty} \Vert \hat{A}^c - A^c \Vert =0, ~\Vert \hat{A}^c - A^c \Vert \sim \mathcal{O}(\frac{1}{\sqrt{l}}).
\end{equation}

The control law learning in infinite case is quite trivial but inspiring, revealing that the feedback gain matrix $K$ can be obtained through OLS-like estimators. In the next subsections, we will turn back to the finite-time problem.

\subsection{Final-state Only Setting}

In this subsection, We study a simplified objective function setting without considering the penalties on states $x_k^T Q x_k$ $(i.e., Q=0)$ described as
\begin{equation}
J_1 = \frac{1}{2} (x_N^T x_N+\sum_{k = 0}^{N-1} u_k^T R u_k).
\end{equation}
We set $H=I$ in the sense that for minimizing the distance between the final state and the target point, the weights on each component of $x_N$ are usually considered equal.
Then, the IOC problem here is to estimate the parameters $R$ only with $M$ optimal state trajectories in the presence of observation noises. We have nothing requirements for data other than Assumption \ref{data_ass}, which means these $M$ trajectories can be just fragments of some complete state trajectories. 

According to the Pontryagin's minimum principle (PMP) \cite{bertsekas2012dynamic}, we introduce the following lemma:
\begin{lemma}\label{lem_pmp}
Consider the optimization problem $\mathbf{P}_0$ with $J_1$.
The optimal control inputs $u_{0:N-1}^*$ and its corresponding state trajectories $x_{0:N}^*$ satisfy
\begin{itemize}
\item[1)]  optimal control policy
$u_i^* = - R^{-1} B^T \lambda_{i+1}^*$,
\item[2)]  costate equation $\lambda_i^* = A^T \lambda_{i+1}^*$,
\item[3)]  terminal condition $\lambda_N^* = H x_N^*$,
\end{itemize}
with the given initial state $x_0^*$, where $\lambda_i$ is the costate of the system, $i = 0,1,\cdots,N-1$.
\end{lemma}

It is straightforward to find that $\mathbf{P}_0$ has the same optimal solution with objective functions $J_1$ and $\alpha J_1, \alpha \in \mathbb{R}_+$. We provide following lemma to reveal that PMP conditions in Lemma \ref{lem_pmp} have this scalar ambiguity property as well. 
\begin{lemma}\label{lem_hr}
Suppose parameters $H, R, x_{1:N}^j,\lambda_{1:N}^j$ satisfy the PMP conditions, then $H'=\alpha H, R' = \alpha R, x_{1:N}^j,\alpha \lambda_{1:N}^j$ for $\alpha \in \mathbb{R}_+$ are also a set of solution to these conditions.
\end{lemma}
\begin{proof}
The proof is given in Appendix \ref{pr1}.
\end{proof}
Note that the PMP condition is the necessary condition for the optimal solution of $\mathbf{P}_0$.
We set PMP condition as the constraint and formulate the inverse control problem Problem \ref{prob_r} based on $M$ trajectory observations $\{\mathcal{Y}^1,\cdots,\mathcal{Y}^M\}$. Since there exist observation noises, $x_{1:N_j}$ and $\lambda_{1:N_j}$ are also optimization variables.
Moreover, in our case we constrain $H=I$, then according to Lemma \ref{lem_hr}, the optimal solution to Problem \ref{prob_r} is supposed to be unique.
\begin{problem}(Inverse Control Problem with PMP)\label{prob_r}
\begin{subequations}
\begin{alignat}{2}\nonumber
&\min_{\hat{R},x_{1:N_j}^j,\lambda_{1:N_j}^j} &&\frac{1}{M}\sum_{j=1}^M \sum_{i = 1}^{N_j} \Vert y_i^j - C x^j_i \Vert^2\\ \nonumber
&~~~~~~\mathrm{s.t.} &&x^j_{i+1} = A x^j_i - B\hat{R}^{-1}B^T \lambda^j_{i+1},\\ \nonumber
&~~~~~~~&& \lambda^j_i = A^T \lambda^j_{i+1}, \, \lambda^j_{N_j} = x^j_{N_j}, \, x^j_0 = y_0^j - \hat{x}_T, \\ \nonumber
&~~~~~~~&& i = 0,1,\cdots,N_j-1, j = 1, \dots,M.
\end{alignat}
\end{subequations}
\end{problem}

We offer the following theorem to further prove the well-posedness property of the problem, which means the inverse problem of $\mathbf{P}_0$ with $J_1$ exists a unique solution. Denote $A^c_k = A - B K_k$ as the close-loop system matrix at time $k$.

\begin{theorem}\label{thm_wp}
Suppose two closed-loop system matrix sequences $A^c_{0:N-1}, {A^c}'_{0:N-1}$ are optimal solutions to problem $\mathbf{P}_0$ with $J_1(R)$ and $J_1(R')$ respectively. If we have $A^c_k = {A^c}'_k$ for all $k$, there is $R = R'$.
\end{theorem}
\begin{proof}
The proof is given in Appendix \ref{pr2}.
\end{proof}

Furthermore, we present Theorem \ref{thm_r} as follow to reveal that the solution of Problem \ref{prob_r} is the true value of $R$ when the amount of observations is large enough.
\begin{theorem}\label{thm_r}
Suppose $\hat{R}, x_{1:N}^*,\lambda_{1:N}^*$ are the optimal solution to Problem \ref{prob_r}. We have
\[
\mathrm{Pr}(\lim_{M \rightarrow \infty} \Vert \hat{R} - R \Vert =0)=1.
\]
\end{theorem}
\begin{proof}
The proof is given in Appendix \ref{pr3}.
\end{proof}
\noindent With the above theorems, we are able to obtain the objective function parameter estimate $\hat{R}$ by solving the Problem \ref{prob_r}.

\subsection{Classic LQR Setting}\label{lqr_setting}
Now we consider adding the process-state term into minimization and the objective function is written as
\begin{equation}\label{hqr_j}
J_0 = \frac{1}{2} x_N^T H x_N+\sum_{k = 0}^{N-1} \frac{1}{2}(x_k^T Q x_k +u_k^T R u_k).
\end{equation}
Solving the optimization problem $\mathbf{P}_0$, we obtain a sequence of time-varying gain matrices $K_{0:N-1}$ calculated by
\begin{equation}\label{dis_k}
K_k = (R + B^T P_{k+1} B)^{-1} B^T P_{k+1} A,
\end{equation}
and $P_k$ satisfies the following iterative equation with the initial value $P_N = H$:
\begin{equation}\label{dis_p}
P_k = K_k^T R K_k + (A - B K_k )^T P_{k+1} (A - B K_k)+Q.
\end{equation}

Inspired by the infinite-time case in Section \ref{inf_sec}, we propose an algorithm consisted of two parts: i) Estimate the feedback gain matrices $\hat{K}_k$ first based on the observation trajectories; ii) Calculate a proper $(\hat{H},\hat{Q},\hat{R})$ with matrices $\hat{K}_k$. We will describe the two parts separately.

\begin{remark}\label{two_alg}
Notice that the form of $J_0$ subsumes $J_1$ ($Q=0,H=I$) in the previous subsection. However, we still divide them and provide two different algorithms since for identifying $J_1$ there is no requirements on observations, which can only be a segment of the trajectory, while in this subsection we need the observation to be a whole trajectory containing the final state. This will be shown in the following part.
\end{remark}

\begin{itemize}[leftmargin=*]
\item {\textbf{Feedback Gain Matrices Estimation}}
\end{itemize}

Suppose we obtain $M$ trajectory observations $\{\mathcal{Y}^1, \dots,$ $\mathcal{Y}^M\}$. As Remark \ref{two_alg} saying, in this subsection we require the observation to contain the final state of each trajectory, which is not difficult to achieve since we have estimate the target state or we can just observe for enough long time waiting for the agent to get to its target. Then we take $l\leqslant \min\{l_1, \dots, l_M\}$ steps from the end of each trajectory and reorder them as 
\[
\bar{\mathcal{Y}}^j = \bar{y}^{j}_{0:l} = y^{j}_{l_j-l:l_j}.
\]
Now we have a truncated trajectory set $\{\bar{\mathcal{Y}}^{1}, \dots, \bar{\mathcal{Y}}^{M}\}$ for subsequent estimation. 
Denote the closed-loop system matrix at time $k$ as $A_k^c = A-BK_k$, then we have
\begin{equation}\label{fin_sys}
y_{k+1} = C A^c_k x_k + \omega_{k+1}
\end{equation}
for all $k$.
Based on equations \eqref{dis_k} and \eqref{dis_p}, we provide the following lemma:
\begin{lemma}\label{k_equal}
If $H,Q,R$ remain unchanged, considering two complete optimal state trajectories $\{x^{1}_{0:N_1}\}, \{x^{2}_{0:N_2}\}$ with different control horizons $N_1\leqslant N_2$ generated by the given system, we have matrix sequence $\{{K_{N_1}^{1}}, K_{N_1-1}^{1}, \dots, K_{0}^{1}\}$ equal to $\{{K_{N_2}^{2}}, K_{N_2-1}^{2}, \dots, K_{N_2-N_1}^{2}\}$.
\end{lemma}

Denote ${\mathcal{X}}^j$ as the state sequence estimated from truncated trajectory $\bar{\mathcal{Y}}^j$.
We utilize following method to estimate states $\hat{x}_k$ from the observation $y_k$ through a filter \cite{yu2021system}.
Denote $D=CB$. For $\{y_{0:l}\}$, there is
\begin{equation}\label{u_est}
\begin{array}{ll}
\zeta_{k+1} = (A-ABD^{\dagger}C) \zeta_k + ABD^{\dagger}(y_k - CA^k x_0), \\
\hat{u}_{k-1} = D^{\dagger} C\zeta_k - D^{\dagger}(y_k - CA^k x_0), ~k = 1,\dots,l,
\end{array}
\end{equation}
where $x_0 = C^{-1}y_0$ and the intermediate variable $\zeta_0 = 0$. Then with $\hat{u}_{0:l-1}$, we have
\begin{equation}\label{x_est}
\begin{array}{ll}
\eta_{k+1} = A \eta_k + AB \hat{u}_{k}, \,
\hat{x}_{k} = -\eta_k - B\hat{u}_{k} + A^k x_0,
\end{array}
\end{equation}
where the intermediate variable $\eta_0 = 0$. We omit the subscript here for brevity, i.e., $\hat{u}_{k}= \hat{u}_{k|l}$ and $\hat{x}_{k}= \hat{x}_{k|l}$.

Then, from Lemma \ref{k_equal}, it is obvious to find that sequence pairs $(\bar{\mathcal{Y}}^{j}, {\mathcal{X}}^j), j=1,\dots,M$ share the same $\{A^c_{0:l-1}\}$. Therefore, similar as the infinite time case, for each $k$, we design the estimator as
\begin{equation}\label{estimator}
\hat{A}^c_k =C^{-1} (Y_k X_k^T)(X_k X_k^T)^{-1},
\end{equation}
where $X_k = (\hat{x}_k^{1}, \dots, \hat{x}_k^{M})^T$ and $Y_k = (\bar{y}_{k+1}^{1} , \dots, \bar{y}_{k+1}^{M})^T$. The matrix $X_k^TX_k$ is required to be full rank ($M\geqslant n$) to guarantee the uniqueness. Then we have
\begin{equation}
\hat{K}_k = B^{\dagger} (A - \hat{A}_k^c),\,k=0,\dots,l-1,
\end{equation}
and $\lim_{l \rightarrow \infty} \Vert \hat{K}_k-K_k \Vert =0$ \cite{richter1995estimating}.
 
\begin{itemize}[leftmargin=*]
    \item {\textbf{Objective Function Parameter Calculation}}
\end{itemize}

After obtaining the feedback gain matrix estimation sequence $\hat{K}_{0:l-1}$, we now find a parameter set $(H,Q,R)$ that generates $\hat{K}_{0:l-1}$ exactly through iteration equations \eqref{dis_k}, \eqref{dis_p}. 
Similarly, we provide the following theorem first to show the scalar ambiguity property under this case.
\begin{theorem}\label{hqr_iden}
Suppose two feedback gain matrix sequences $K_{0:N-1}, K'_{0:N-1}$ are generated with two sets of parameters ${H},{Q},{R}$ and ${H'},{Q'},{R'}$ respectively through equation \eqref{dis_k}. If there exist at least $\frac{mn(n+1)(m+1)}{2}$ linearly independent $vec(\mathcal{P}_i(\mathcal{E}_i))$ defined in \eqref{pie} and $K_k=K'_k$ for all $k$, we have ${H'} = \alpha H, {Q'} = \alpha Q, {R'}=\alpha R$ for some $\alpha \in \mathbb{R}_+$.
\end{theorem}
\begin{proof}
See the proof in Appendix \ref{hqr_iden_pr}.
\end{proof}

\begin{corollary}
Theorem \ref{hqr_iden} provide a criteria for the identifiability of the objective function. If the control horizon is set as $N < \frac{mn(n+1)(m+1)}{2}$, the true weight parameters $H,Q,R$ of the control objective will never be identified accurately, which can be utilized in preserving the system's intention.
\end{corollary}

Now, with the identifiability guarantee, we introduce our IOC algorithm based on following lemma derived from \eqref{dis_k}.
\begin{lemma}\label{lem_rq_eq}
For $H,Q,R$ in the objective function $J_0$, we have
\begin{equation}\label{rq_eq}
a_i (I_i \otimes R) b_i = c_i \begin{bmatrix} H &0 \\0 & I_{i-1} \otimes Q \end{bmatrix} d_i,
\end{equation}
for $\, i=1,\dots,T$ and
\begin{equation}\label{abcd}
\begin{aligned}
& a_i = [I_n ~~ -B^T K_{N-i+1}^T ~~ -B^T {A_{N-i+1}^c}^T K_{N-i+2}^T ~~ \cdots \\& ~~~~~~ -B^T \prod_{r = 2}^{i-1} {A_{N-r}^c}^T K_{N-1}^T],\\
& c_i = B^T \begin{bmatrix}
\prod_{r = 1}^{i-1} A_{N-r}^c & \prod_{r = 2}^{i-1} A_{N-r}^c & \cdots & A_{N-i+1}^c & I_n
\end{bmatrix},\\
& b_i = \begin{bmatrix}
K_{N-i} \\ K_{N-i+1} A_{N-i}^c \\ \vdots \\ K_{N-1} \prod_{r = 2}^{i} A_{N-r}^c
\end{bmatrix},\, 
d_i = \begin{bmatrix}
\prod_{r = 1}^{i} A_{N-r}^c \\ \prod_{r = 2}^{i} A_{N-r}^c \\ \vdots \\ A_{N-i}^c
\end{bmatrix},
\end{aligned}
\end{equation}
where $K_{N-T:N-1}$ are gain matrices.
\end{lemma}
\begin{proof}
See the proof in Appendix \ref{pr_lem_rq}.
\end{proof}
\noindent Note that equation \eqref{rq_eq} iterates from $k=N-1$, which is the reason why we require the observations to contain the final state of trajectories in the previous step. 

Set \eqref{rq_eq} as the constraint of our estimation problem and Theorem \ref{hqr_iden} ensures the identifiability of $H,Q,R$ (i.e., as the observations increase, the estimation error of $\hat{K}_k$ decreases and we obtain the real parameters), while we use an additional criteria to guarantee the uniqueness of the solution that estimates must minimize the condition number of the block diagonal matrix consisting of $\hat{H},\hat{Q},\hat{R}$. Supposing we obtain $\hat{K}_{N-T:N-1}$, the optimization problem is formulated as:
\begin{problem}($H,Q,R$ Estimation with Condition Number Minimization)\label{qr_pro}
\begin{equation}\label{lmi}
\begin{aligned}
&(\hat{H}, \hat{Q}, \hat{R}, \hat{\tau}) = \arg \min_{H,Q,R,\tau} \tau^2\\
&~ \mathrm{s.t.} ~~ \eqref{rq_eq}, \, I \preceq \mathrm{diag}(H,Q,R) \preceq \tau I.
\end{aligned}
\end{equation}
\end{problem}
For Problem \ref{qr_pro}, the number of equality constraints and $T\leqslant N$ in \eqref{rq_eq} can be decided by the trade-off between accuracy and computation cost. If the gain matrix estimations are all accurate, the above LMI problem is feasible. Since it is a convex optimization problem with linear constraints, there exists at least one exact solution. However, due to the existence of observation noises, the estimation may contain errors and the problem has no solution (infeasible).
Therefore, We offer a further analysis to determine whether Problem \ref{qr_pro} is solvable. Write formula \eqref{rq_eq} into following expression:
\[
(b_i^T \otimes a_i) \cdot \underbrace{\vphantom{\begin{bmatrix} H & \\ & I_{i-1} \otimes Q \end{bmatrix}}vec(I_i \otimes R)}_{\mathcal{G}^1(R)} = (d_i^T \otimes c_i) \cdot \underbrace{vec(\begin{bmatrix} H & 0\\0 & I_{i-1} \otimes Q \end{bmatrix})}_{\mathcal{G}^2(H,Q)},
\]
for $i = 1,\dots,T$. Define the rows of zero element in the vector $\mathcal{G}^1(R)$ as set $\mathcal{C}^1$, $\mathcal{C}^1 = \{k|\mathcal{G}^1(R)_k = 0\}$ and set $\mathcal{C}^2$ for $\mathcal{G}^2(H,Q)$ similarly. Then we take
\[
\mathcal{K}_i^1 = (b_i^T \otimes a_i)(:,\mathcal{C}^1),\, \mathcal{K}_i^2 = (d_i^T \otimes c_i)(:,\mathcal{C}^2),
\]
with which \eqref{rq_eq} is simplified into
\begin{equation}
\mathcal{K}_i^1 \cdot (\bm{1}_i \otimes vec(R)) = \mathcal{K}_i^2 \cdot \begin{bmatrix}vec(H) \\ \bm{1}_{i-1} \otimes vec(Q)\end{bmatrix}
\end{equation}
Therefore, we can combine all the $T$ equations as
\begin{equation}
\begin{aligned}
&\Tilde{\mathcal{K}}_T^1 \cdot (\bm{1}_T \otimes vec(R)) =\Tilde{\mathcal{K}}_T^2 \cdot \begin{bmatrix}vec(H) \\ \bm{1}_{T-1} \otimes vec(Q)\end{bmatrix}\\
&\Leftrightarrow \underbrace{\vphantom{\begin{bmatrix}1\\1\\1\end{bmatrix}}\begin{bmatrix}\Tilde{\mathcal{K}}_T^1 & -\Tilde{\mathcal{K}}_T^2\end{bmatrix}}_{\Phi_T} \cdot \underbrace{\begin{bmatrix}\bm{1}_T \otimes vec(R) \\ vec(H) \\ \bm{1}_{T-1} \otimes vec(Q)\end{bmatrix}}_{\Theta_T(H,Q,R)} = \begin{bmatrix}
0\\ \vdots \\ 0
\end{bmatrix},
\end{aligned}
\end{equation}
where
\[
\Tilde{\mathcal{K}}_T^1 = \begin{bmatrix}
\mathcal{K}_1^1 & 0 & \cdots & 0\\
\multicolumn{2}{c}{\mathcal{K}_2^1}& \cdots & 0\\
\multicolumn{2}{c}{\ddots} & \vdots &\vdots \\ 
\multicolumn{4}{c}{\mathcal{K}_T^1}
\end{bmatrix}, 
\Tilde{\mathcal{K}}_T^2 =\begin{bmatrix}
\mathcal{K}_1^2 & 0 & \cdots & 0\\
\multicolumn{2}{c}{\mathcal{K}_2^2}& \cdots & 0\\
\multicolumn{2}{c}{\ddots} & \vdots &\vdots \\ 
\multicolumn{4}{c}{\mathcal{K}_T^2} 
\end{bmatrix}.
\]

Based on the above derivation, we provide the following theorem.
\begin{theorem}\label{infeasible}
Problem \ref{qr_pro} is infeasible, when the rank
\begin{equation}
rank(\Phi_T) \geqslant n^2+n+\frac{m^2+m}{2},
\end{equation}
which implies that the equation
\[
\Phi_T \cdot \Theta_T(X_h,X_q,X_r) = \bm{0}
\]
has only zero solution for unknown variables $X_h,X_r\in \mathbb{S}^{n}$ and $X_q\in \mathbb{S}^{m}$.
\end{theorem}

Therefore, we consider the following optimization problem instead when Problem \ref{qr_pro} is infeasible:
\begin{equation}\label{qp_pro}
\begin{aligned}
& \min_{\hat{H}, \hat{Q}, \hat{R}} ~\Vert \Phi_T  \cdot \Theta_T(\hat{H}, \hat{Q}, \hat{R})\Vert_2^2 \\
& ~~s.t.~~~\Vert \Theta_T(\hat{H}, \hat{Q}, \hat{R}) \Vert = 1,
\end{aligned}
\end{equation}
which can be solved by existing QP solvers \cite{stellato2020osqp}.

\section{Control Horizon Estimation and LQR Problem Reconstruction}\label{horizon}
Note that with the target state and objective function, we can generate a satisfying control trajectory. However, in order to imitate the agent's motion precisely, we still need to estimate the control horizon, since different control horizons lead to different inputs and final states. 

In this section, we will first introduce the control horizon estimation algorithm, then reconstruct the LQR problem with estimates and provide the input prediction application.
\subsection{Control Horizon Estimation}
Suppose $R_m$ is now driving under an optimal trajectory generated by $\mathbf{P}_0$ toward the target. To estimate the control horizon $N$ of this trajectory, we need an $l$ length continuous observation 
\begin{equation}\label{y_observe}
\mathcal{Y} = \{y_0, y_1,\dots, y_l\}.
\end{equation}
We build the following optimization problem for estimation.
\begin{problem}(Estimation of the Control Horizon)\label{pro_n}
\begin{subequations}
\begin{eqnarray}\nonumber
&&\min_{\hat{N}} ~~\sum_{i = 1}^l \Vert y_i - C x_i \Vert^2 : = J_N(\hat{N};y_{0:l})\\ \nonumber
&&~\mathrm{s.t.} ~~ {x}_{k+1} = A {x}_k + B {u}_k,\\ \nonumber &&~~~~~~~ \,u_i = -K_i x_i,\,x_0 = y_0- \hat{x}_T,\\ \nonumber
&&~~~~~~~ K_i =  (\hat{R} + B^T P_{i+1} B)^{-1} B^T P_{i+1} A, \\ \nonumber
&&~~~~~~~ P_i =K_i^T \hat{R} K_i + {A^c_i}^T P_{i+1} A^c_i+\hat{Q}, \,P_{\hat{N}} = \hat{H},\\\nonumber
&&~~~~~~~
i = 0,1,\dots,\hat{N}-1,
\end{eqnarray}
\end{subequations}
where $y_{1:l}$ is the observation of $R_m$ up to time $k=l$.
\end{problem}
The above problem reflects that the deviation of the trajectory obtained under the optimal solution $\hat{N}^*$ from the observed data $y_{1:l}$ is the smallest. Since the optimization variable $\hat{N} \in \mathbb{N}_+$ does not explicitly exist in the objective and constraints, the problem is a non-convex optimization on the set of positive integers and hard to solve directly. Therefore, we turn to investigate how the value of the objective function changes with $\hat{N}$. Note that $N > l$ and as $\hat{N}$ grows from $l$ to the real horizon $N$, the function $J_N$ gradually decreases to the minimum point. Then, as the continue growth of $\hat{N}$, $J_N$ increases and finally converges to a fixed value $\sum_{i = 1}^l \Vert y_i\Vert^2$. See an illustration in Fig. \ref{n_est} of Section \ref{sim}.

According to the analysis of $J_N$, to obtain the solution of Problem \ref{pro_n}, we can start traversing from $\hat{N} = l$ and keep increasing $\hat{N}$ until $J_N$ no longer decreases, at which time the optimal $\hat{N}^*$ is found. However, this will consume lots of computation if $N \gg l$. Therefore, we propose an algorithm based on binary search to find the optimal solution inspired by the line search of gradient decent method. Since $J_N$ is a discrete function with respect to $\hat{N}$, we use the function values at both $\hat{N}$ and $\hat{N}+1$ to approximate the gradient at point $\hat{N}$, which is given by:
\begin{equation}\label{gradient}
g_N = \frac{J_N(\hat{N}+1;y_{0:l})-J_N(\hat{N};y_{0:l})}{(\hat{N}+1)-\hat{N}}.
\end{equation}
Thus, if we have $g_N < 0$, then $\hat{N} < \hat{N}^*$; if $g_N>0$, then $\hat{N} \geqslant \hat{N}^*$.
The detailed algorithm is shown in Algorithm \ref{alg1}.

\begin{algorithm2e}\label{alg1}
\SetAlgoLined
 \caption{Binary Search for Optimal $\hat{N}$}
    \KwIn{ 
    The observation trajectory and the observation times, $y_{0:l}, l$;
      The estimate of the target state, $\hat{x}_T$;
      The parameters of the system dynamic, $A,B,C$;
      The estimate of the objective function parameter, $\hat{H},\hat{Q},\hat{R}$; The step length, $\theta$;}
    \KwOut{
      The optimal control horizon estimate, $\hat{N}^*$;}
    \textbf{Determine the initial bound:}\\
    Set the lower bound as $N^{-} = l+1$;
    Let $\hat{N}' = l + \theta$;\\
    \lWhile{$g_{\hat{N}'} < 0$}{
    $\hat{N}' = \hat{N}' + \theta$;
    }
    Set the upper bound as $N^+ = \hat{N}'$;\\
    \textbf{Binary search:}\\
    Take the midpoint of the range as $\Tilde{N} = \lfloor{\frac{N^- + N^+}{2}}\rfloor$;\\
    \While{$N^+ - N^- > 1$}{
    \lIf{$g_{\Tilde{N}} > 0$} {Set $N^+ = \Tilde{N}$}
    \lElse{Set $N^- = \Tilde{N}$}
    }
    $\hat{N}^* = \mathop{\arg\min}\limits_{\hat{N}} \{J_N(\hat{N};y_{0:l}); \hat{N} \in \{N^+,N^-\} \}$;\\
\textbf{return} The control horizon estimation $\hat{N}^*$.
\end{algorithm2e} 

\begin{remark}
Notice that after determining the initial range $[N^-,N^+]$, if we simply leverage the function values of the two break points $N_1,N_2,N^-\leqslant N_1<N_2 \leqslant N^+$ to find the optimum, to have a constant compressive ratio $c$, the break points should satisfy $\frac{N_2-N^-}{N^+ - N^-}= \frac{N^+ - N_1}{N^+ - N^-} = c$ and $c = \frac{\sqrt{5}-1}{2} \approx 0.618$. However, with the above Algorithm \ref{alg1}, since we approximate the gradient, the compressive ratio is improved to $c = 0.5$. Moreover, in order to reduce the computation cost, it is recommended to store the result each time we calculate the deviation sum $J_N$ corresponding to a certain $\hat{N}$ to avoid repeated calculation.
\end{remark}

\subsection{LQR Reconstruction}
Now we have obtained the estimate of the target state $\hat{x}_T$, the control horizon $\hat{N}^*$, and identified the weighting matrices in objective function $J_0$. Therefore, we can reconstruct the optimization problem by substituting $x_T,H,Q,R,N$ in $\mathbf{P}_0$ with our estimates and calculate the control law $\hat{K}_{0:\hat{N}-1}$ through the iteration \eqref{dis_k} and \eqref{dis_p}.

We provide the future input prediction as one important application of our control law learning algorithm.
Now suppose at time $k = l$, the agent $R_o$ has observed $R_m$ for $l+1$ steps and obtained a series of observations $\mathcal{Y}$ as \eqref{y_observe}.
Denote $u_l$ as the real input at $k=l$ that we want to predict and $\hat{u}_{l|l}$ as our input inference. $R_o$ tries to infer the current control input accurately, which is to minimize the error between $\hat{u}_{l|l}$ and $u_l$. From \eqref{eq:uk}, $u_l$ is calculated by $-K_l x_l$, then we have
\begin{equation}
\min \Vert u_l - \hat{u}_{l|l} \Vert^2 = \Vert -K_l x_l - \hat{u}_{l|l}\Vert^2,
\end{equation}
Notice that the state estimate $\hat{x}_{l|l}$ can be obtained by Kalman filter. 
Therefore, we can infer the control input of the target agent $R_m$ at time $l$ by 
\begin{equation}
\hat{{u}}_{l|l} = \mu_0,
\end{equation}
where $\mu_0$ is calculated by solving the following reconstructed optimization problem:
\begin{problem}(Control Input Prediction)\label{p1}
\begin{subequations}
\begin{eqnarray}\nonumber
&&\min_{\mu_{0:N'-1}} J' = x_{N'}^T \hat{H} x_{N'} + \sum_{k = 0}^{N'-1} (x_k^T \hat{Q} x_k+\mu_k^T \hat{R} \mu_k)\\ 
&&~~~\mathrm{s.t.} ~~\eqref{sys}, x_0 = \hat{x}_{l|l} - \hat{x}_T,\nonumber
\end{eqnarray}
\end{subequations}
in which $N' = \hat{N}^*-l$.
\end{problem}
\noindent The solution of this reconstructed LQR problem is given in \cite{kwakernaak1972maximally} as
\begin{equation}\label{mu0}
\mu_0 = - K_0 x_0,
\end{equation}
where $K_0$ is calculated by \eqref{dis_k}.
\begin{remark}
According to the principle of optimality \cite{bellman1966dynamic} in the dynamic programming, if a control policy $p_{0,N}^*$ is optimal for the initial point $x_0$, then for any $l\in \{1,2,\cdots,N-1\}$, its sub-policy $p_{l,N}^*$ is also optimal to the subprocess containing last $N-l+1$ steps with the initial point $x_l$.
\end{remark}

\begin{itemize}[leftmargin=*]
\item {\textbf{Prediction Error Analysis}}
\end{itemize}

Note that when the trajectory sample size $M \rightarrow \infty$, we have $\hat{x}_T = x_T$ and $\hat{H} = H, \hat{R} = R,\hat{Q}=Q$ according to the law of large numbers. However, the accuracy of estimation $\hat{N}^*$ can not be guaranteed, since it is the value that minimize function $J_N$, which is affected by the observation errors. Therefore, the sensitivity of $\hat{N}^*$ with respect to the input estimation $\hat{{u}}_{l|l}$ needs to be analyzed. 

It is found that the influence of $\hat{N}^*$ is reflected in the calculation of $K_0$ with Problem \ref{p1} and formula \eqref{mu0}.
Now suppose the real control horizon of the system $N$ generates $K_0^r$, while our estimate $\hat{N}^*$ calculates $\hat{K}_0$. We will show the estimation error $\Vert \hat{K}_0 - K_{0}^{r} \Vert$ can be bounded and controlled in the following analysis.

Note that a discrete-time LQR problem with finite control horizon as $\mathbf{P}_0$ is solved through dynamic programming method and the iteration equation of the intermediate parameter $P_k$ can be described as
\begin{equation}\nonumber
P_{k-1} = A^TP_{k}A - A^T P_{k} B(R+ B^T P_{k}B)^{-1} B^TP_{k}A+Q,
\end{equation}
for $k=1,\dots,N$ with $P_N = H>0$. According to \cite{bitmead1991riccati}, when $N \xrightarrow{} \infty$, there is $P_0=P^*>0$, where $P^*$ satisfies the discrete Riccati equation:
\begin{equation}
P^* = A^TP^*A - A^T P^* B(R+ B^T P^*B)^{-1} B^TP^*A+Q,
\end{equation}
and correspondingly,
\[
K_0 =K^*= ({R} + B^T P^* B)^{-1} B^T P^* A.
\]
What's more, the sequence $\{P_k\}$ is monotonic (In our analysis, a matrix sequence is monotonic means $P_0  \lesseqgtr P_1 \lesseqgtr \dots \lesseqgtr P_N$, where $P_i \geqslant P_j$ implies $P_i - P_j$ is a positive semi-definite matrix).
According to Lemma \ref{k_equal}, the difference between the ends of two trajectories with control horizons $N_1,N_2, N_1\leqslant N_2$ can be converted to a comparison inside a single sequence generated by $N_2$, which is $\Vert K_{0}^{(1)} - K_{0}^{(2)} \Vert = \Vert K_{N_2-N_1}^{(2)} - K_{0}^{(2)} \Vert$.
Therefore, we have the estimation error 
\[
\Vert \hat{K}_0 - K_{0}^{r} \Vert = \Vert K_{|\hat{N}^* - N|} - K_{0} \Vert, 
\]
where $K_{0:N_m}$ is generated under the horizon $N_m = \max(N,\hat{N}^*)$. Then we will focus on the convergence of the $\{K_k\}$ sequence under fixed $H,Q,R$.

We offer the following theorem to show the input prediction error and sensitivity of the control horizon estimate.

\begin{theorem}\label{thm_sens}
There exists a positive integer $\overline{N}\in \mathbb{N}_+$ and $\eta > 0$, which can be set as the maximum tolerable inference error. For any $N> \overline{N}$, we have $\Vert \mu_0^{(N+\delta N)} - \mu_0^{(N)} \Vert \leqslant \eta$, where $\mu_0^{(N)}$ denotes the input inference when the control horizon estimate $\hat{N}^*=N$ and $\eta$ is proportional to $\delta N \in \mathbb{N}_+ $.
\end{theorem}
\begin{proof}
See the proof in Appendix \ref{prf_sens}.
\end{proof}

The complete algorithm flow is shown in Algorithm \ref{alg2}.

\begin{algorithm2e}\label{alg2}
 \caption{LQR Reconstruction based Control Inputs Prediction Algorithm}
    \KwIn{ The observation data including $M$ history trajectories, $\{\mathcal{Y}^{1:M}\}$;
    The observation of current trajectory and the observation times, $\mathcal{Y}, l$;
    The system dynamic function, $A,B,C$;}
    \KwOut{
      The control input inference at time $l$, $\hat{u}_{l|l}$;}
    Estimate the target state $\hat{x}_T$ through curve fitting and calculating the intersection;\\
    \uIf{consider only final-state}
    {Identify the objective function parameter $\hat{R}$ by solving Problem \ref{prob_r} with $M$ trajectories;}
    \uElseIf{consider process-states}{
    Estimate the feedback gain matrices;\\
    Identify $\hat{H},\hat{Q},\hat{R}$ by solving Problem \ref{qr_pro};}
    \lElse{Return false}
    Calculate the optimal estimate to control horizon $\hat{N}^*$ with Algorithm \ref{alg1};\\
    Formulate and solve Problem \ref{p1} with previous estimates; Obtain the one-step input $\mu_0$ with \eqref{mu0} in the forward pass;\\
    \textbf{return} The control input prediction $\hat{u}_{l|l} = \mu_0$.
\end{algorithm2e}

\section{Simulation Results}\label{sim}
In this section, we conduct multiple simulations on our algorithm and apply it to the future input and trajectory prediction to show the performance and efficiency.

Consider a controllable linear system modeled by a three-dimensional dynamical function as \eqref{sys} and
\begin{equation}\label{sim_sys}
\begin{aligned}
& A = \begin{bmatrix}
    1.4155 &  -0.0876 &   0.7213\\
    0.8186 &   2.7338 &  -1.2750\\
   -0.3118 &  -0.7573 &  1.2008
\end{bmatrix} ,\\
& B = \begin{bmatrix}
    -0.0484 &   0.1611 &  -1.8972\\
   -1.1350 &   1.6600  &  0.1003\\
    0.3905 &  -0.7851  &  0.1055\\
\end{bmatrix},
\end{aligned}
\end{equation}
are generated randomly where $n = m =3$.
Assume the agent is driving to the target state $\begin{bmatrix}
    6,8,4
\end{bmatrix}^T$. The observation function is \eqref{ob_fun} with $C=I_3$
and the observation noise satisfies a Gaussian distribution $\mathcal{N}(0,0.02^2)$.
We now obtain a set of trajectory observation. Through applying the external incentives and the line fitting, we calculate the intersection point as the estimation to the target state \[\hat{x}_T = \begin{bmatrix}
    6.063,8.086,4.039
\end{bmatrix}^T.\]

\subsection{Final-state Only Setting}

Firstly, we consider the control optimization problem $\mathbf{P}_0$ with final-state setting (as $J_1$). We suppose
\[
H = I_{3\times 3}, R = \begin{bmatrix}
    0.4 I_{2\times 2} & 0_{2\times 1} \\
    0_{1\times 2} & 0.8
\end{bmatrix}.
\]
To test the estimation algorithm for $R$, we set the trajectory length $l_j=10$ and random initial states $x_0^j$ for all $j = 1,2,\cdots,M$ to ensure the linear independence. Use the MATLAB function $fmincon$ with ``interior-point" method to solve Problem \ref{prob_r}.
We pick the Frobenius norm to measure the estimation error:
\[
err(\hat{R}) = \frac{\Vert \hat{R}-R \Vert_F}{\Vert R \Vert_F}.
\]
The results are shown in Fig. \ref{r_est} and the time costs are listed in Table.\ref{tab1}. Notice that $R$ is a three-dimensional matrix, more than three trajectories are required to solve for a unique $\hat{R}$. We can see that as the number of trajectories $M$ increases from $3$ to $13$, the estimation error shows a decreasing trend. However, as $M$ becomes larger, the search space of the problem increases as well and we need to continuously enlarge the parameter $MaxfunEvals$ of function $fmincon$ to ensure an accurate solution. The lager number of iterations leads to a longer solving time. Therefore, considering the trade-off between the estimation error and the computational efficiency, we choose $M=7$ and the estimation value 
\[
\hat{R} = \begin{bmatrix}
    0.4193 &  -0.0078 &  -0.0063\\
   -0.0026 &   0.4083 &  -0.0071\\
   -0.0089 &  -0.0111  &  0.8166\\
\end{bmatrix}.
\]

\begin{figure}[t]
\centering
\includegraphics[width=0.35\textwidth]{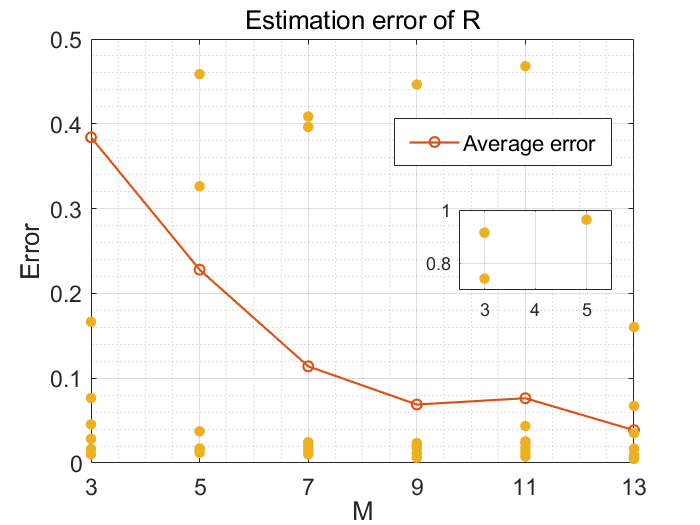}
\caption{Estimation results of the parameter $R$ under only final-state setting. We launch the simulation $8$ times at each $M$. The yellow dots are estimation errors of each simulation and the red curve represents the change of the average error from $M=3$ to $13$.}
\label{r_est}
\end{figure}

\begin{table}
    \centering
    \caption{Time cost of $R$ estimation}
	\label{tab1}
    \begin{tabular}{cccccc}
    \toprule    
    Amount of Data $M$ & 2 & 4 & 6 & 8 & 10 \\    
    \midrule   
    Time Cost (s) & 0.5436 & 0.7829 & 1.071 & 1.431 & 1.806\\
    \bottomrule   
    \end{tabular}
\end{table}

\subsection{Classic LQR Setting}

\begin{figure}[t]
\centering
\includegraphics[width=0.35\textwidth]{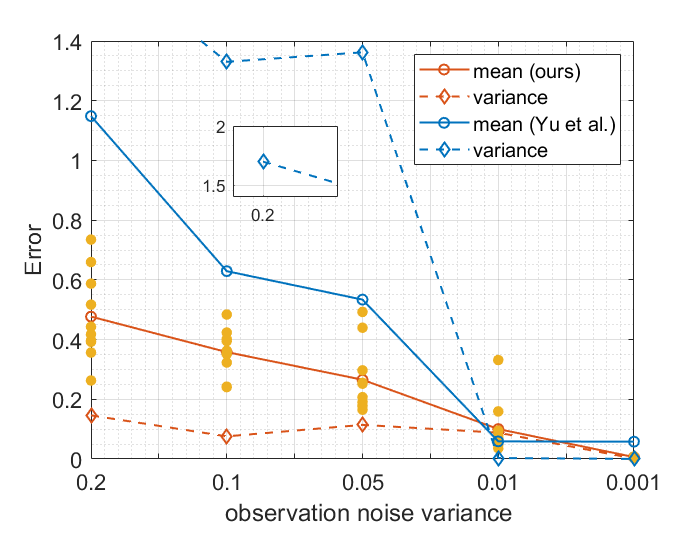}
\vspace{-4pt}
\caption{Estimation results of the parameters $H,Q,R$. We run experiments for $10$ times at different variances of observation noise from $0.2$ to $0.001$. The red lines show estimation errors of our algorithm with yellow dots representing each simulation results and blue lines are results of method proposed in \cite{yu2021system} ($H=Q$), where the solid line is mean of the errors and the dashed line is the standard deviation.}
\label{hqr_est}
\end{figure}
Now we consider the more complex objective function setting as $J_0$ in \eqref{hqr_j}. We set parameters
\[
H = I_{3\times 3},\, Q = 0.2 I_{3\times 3} ,\,R = \begin{bmatrix}
    0.4 I_{2\times 2} & 0_{2\times 1} \\
    0_{1\times 2} & 0.8
\end{bmatrix}
\]
and estimate them simultaneously by the algorithm in Sec. \ref{ioc}-B. We first collect $M$ trajectories containing their terminal states and compute the feedback matrix $K_k$ at each step using the proposed estimator as \eqref{estimator}. Note that when there is no observation noise, Problem \ref{qr_pro} is always feasible with any $T$ since $rank(\Phi_T) = 9 < n^2+n+\frac{m^2+m}{2} = 18$ according to Theorem \ref{infeasible}, and the solutions are accurately equal to the real $H,Q,R$. However, if the observation noise exists, in this case Problem \ref{qr_pro} is only feasible at $T=1$ ($rank(\Phi_1) = 9<18$). For $T \geqslant 2$, we have $rank(\Phi_1) = 9 \cdot T \geqslant 18$ which needs to transfer into the QP problem as \eqref{qp_pro}. 

Here we set $T = 6$ and solve Problem \ref{qr_pro} and \eqref{qp_pro} with solver YALMIP \cite{lofberg2004yalmip} and SeDuMi \cite{sturm1999using} in MATLAB. The estimation errors are shown in Fig. \ref{hqr_est}. We use the Frobenius norm to measure the estimation error:
\[
err = \frac{\Vert \begin{bmatrix}\hat{H}&\hat{Q}&\hat{R}\end{bmatrix}-\begin{bmatrix}H & Q & R\end{bmatrix} \cdot \alpha \Vert_F}{\Vert \begin{bmatrix}H & Q & R\end{bmatrix} \cdot \alpha \Vert_F}.
\]
With the trajectory number fixed, as observation noises decrease, the estimate to $\hat{K}_k$ becomes more accurate and estimation errors of $\hat{H},\hat{Q},\hat{R}$ gradually decrease. 
We compare with the IOC algorithm proposed in \cite{yu2021system} (we set $H=Q$ here since \cite{yu2021system} estimates $Q,R$ only) and find that when the observation noise is small enough, the estimation errors obtained by two algorithms are basically the same, while when the noise is not negligible, ours works better. What's more, comparing the variance of multiple simulation results, our algorithm is overall smaller and more stable than theirs.
Finally we take the following estimates:
\begin{equation}\nonumber
\begin{aligned}
&\hat{H} = \begin{bmatrix}
5.0699 &  -0.1320 &  -0.1507\\
   -0.1320 &   4.9445  & -0.0058\\
   -0.1507 &  -0.0058  &  4.9836
\end{bmatrix}, \\
&\hat{Q} = \begin{bmatrix}
1.0069  &  0.0097  &  0.0045\\
    0.0097 &   1.0137  & 0.0064\\
    0.0045  &  0.0064 &   1.0030
\end{bmatrix}, \\
&\hat{R} = \begin{bmatrix}
1.9926 &  -0.0791  & -0.1309\\
   -0.0791  &  1.9727  & -0.0556\\
   -0.1309 &  -0.0556  &  3.8733
\end{bmatrix}
\end{aligned}
\end{equation}
with the multiplied scalar $\alpha = 5$.

Now we start to infer the inputs and predict the future states of the agent. Observe for $l =15$ times and set $\theta = 10$. With Algorithm \ref{alg1} and the property of $J_N$ curve in Fig. \ref{n_est}, we obtain the optimal control horizon estimation $\hat{N}^* = 20$. At this point, we have completed the estimation of $\hat{x}_T, \hat{R}$ and $\hat{N}^*$. The current state is estimated by Kalman filter as $\hat{x}_{l|l}$. Then, we reconstruct and solve the control optimization problem of the mobile agent as Problem \ref{p1}:
\begin{equation}\nonumber
\begin{aligned}
&\min_{\mu_{0:4}}~~ x_{5}^T \hat{H} x_{5} + \sum_{k = 0}^{4} (\mu_k^T \hat{R} \mu_k + x_k^T \hat{Q} x_k) \\ &~~\mathrm{s.t.} ~~\eqref{sys}, x_0 = \hat{x}_{l|l} - \hat{x}_T, k = 0,\dots,4.
\end{aligned}
\end{equation}
Note that the solution $x_{1:5}$ of above problem are actually the prediction to the future state of the agent from $k=16$ to $20$. Denote the prediction error as $\Vert \hat{x} - x\Vert$. Comparing with the prediction through polynomial regression \cite{chen2016tracking} in presence of the same observation noise distribution, the results are shown in Table.\ref{tab2}. The curve fitting is based on all the $l=15$ history states and the highest polynomial order is chosen as $3$ which is the optimal. We can see that the prediction error generated by our methods is overall lower than the fitting methods. Moreover, the error by curve fitting grows larger with time $k$ while the error of our method decreases as $x_k$ goes to $0$.

\begin{figure}[t]
\centering
\includegraphics[width=0.35\textwidth]{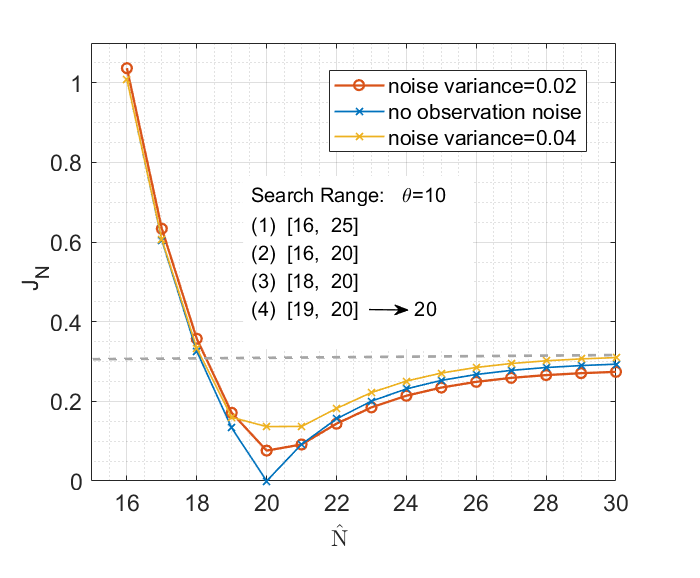}
\vspace{-6pt}
\caption{The function value $J_N$ of different estimation $\hat{N}$. Set the observation number $l = 15$ and the real control horizon $N=20$. The red curve shows the change of function $J_N$ with $\hat{N}$ when the observation noise variance equal to $0.02$ (our simulation setting). We can see that $J_N$ decreases at the beginning and reaches the minimum value at $20$. Then it gradually increases and converges to $\sum_{i = 1}^{15} \Vert y_i\Vert^2$. Note that as the noise variance goes larger, the minimum point of $J_N$ may deviate from $20$ as the yellow curve and Theorem \ref{thm_sens} is proposed for analyzing the error's bound.}
\label{n_est}
\end{figure}

\begin{table}
    \centering
    \caption{States prediction error comparison}
	\label{tab2}
    \begin{tabular}{cccccc}
    \toprule    
    Step $k$ & 16 & 17 & 18 & 19 & 20 \\    
    \midrule   
    \textbf{Ours} & 0.0365 & 0.0271 & 0.0184 & 0.0106 & 0.0048\\
    \cite{chen2016tracking} & 0.0543 & 0.1105 & 0.2006 & 0.3378 & 0.5401\\
    \bottomrule   
    \end{tabular}
\end{table}

\section{Experiments}\label{exp}
\begin{figure}[t]
\centering
\includegraphics[width=0.32\textwidth]{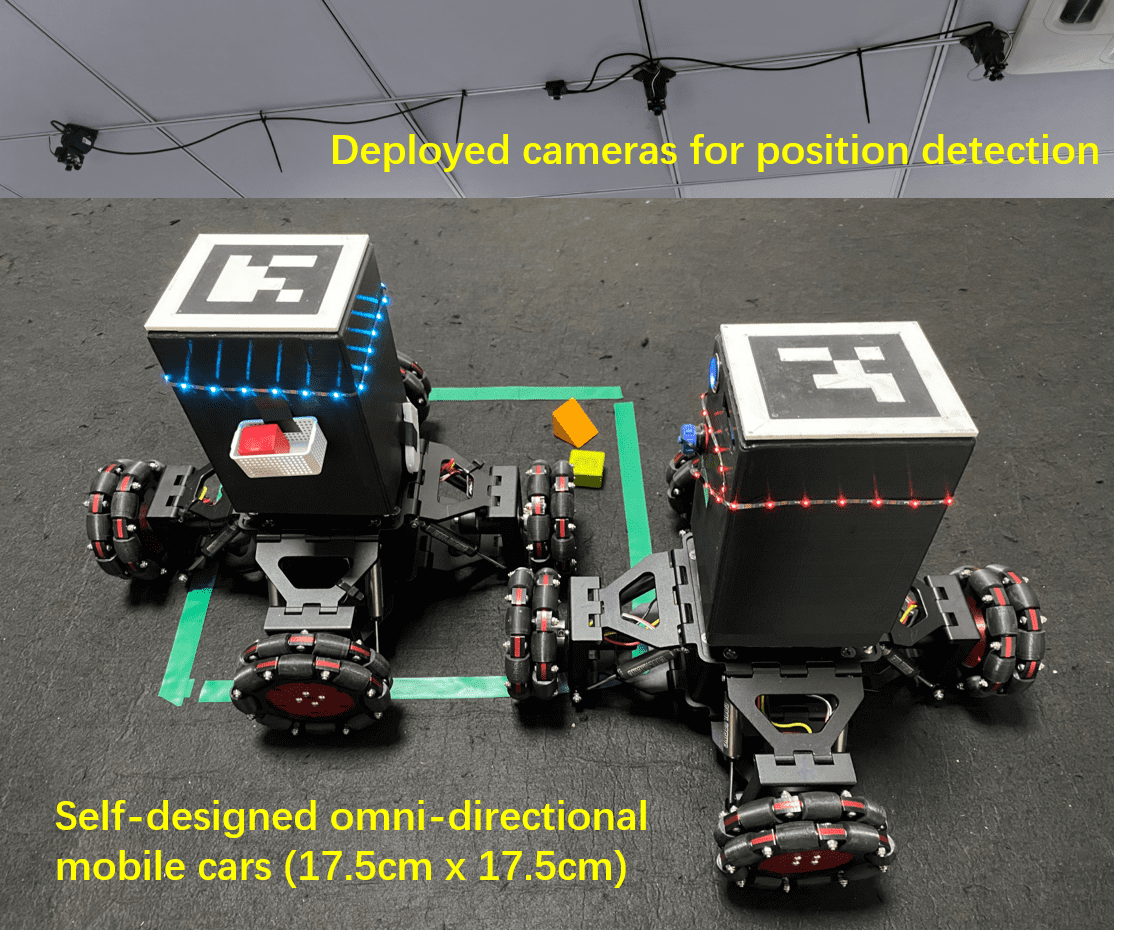}
\vspace{9pt}
\caption{The self-designed platform. Each omni-directional mobile agent has a unique QR code on its top for the position detection. A high-precision camera array is deployed on the ceiling.}
\label{plat}
\end{figure}

\begin{figure}[t]
\centering
\includegraphics[width=0.45\textwidth]{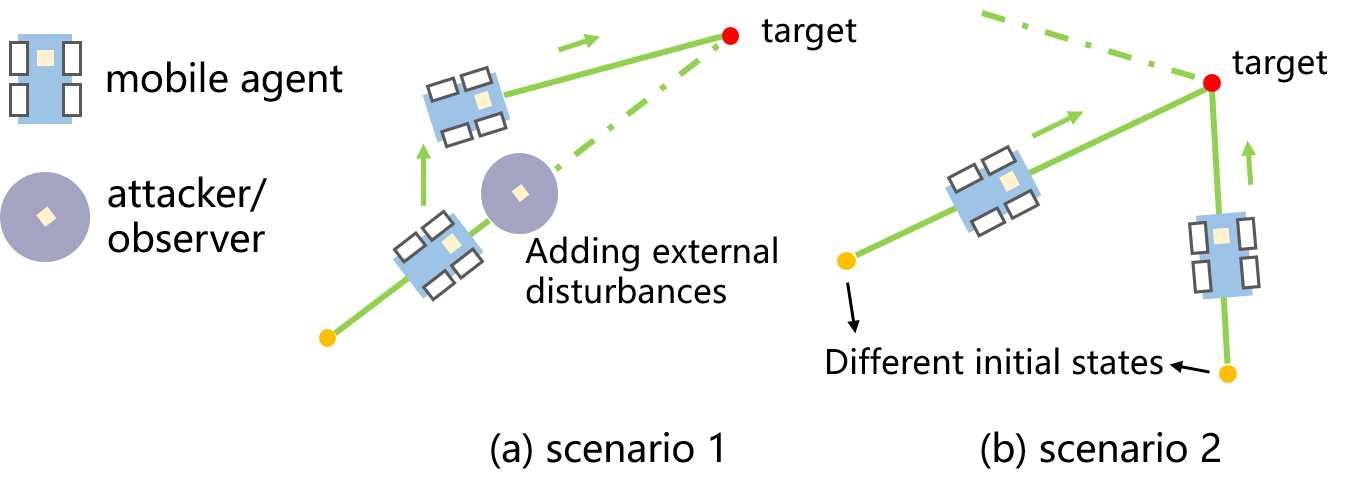}
\caption{Two scenarios of data collection process. Scenario 1 leverages the mechanism that the agent will re-calculate the trajectory after receiving disturbances, while scenario 2 just observes and records different trajectories from different initial point.}
\label{data_col}
\end{figure}
We demonstrate the algorithm on our self-designed mobile robot platform \cite{ding2021robopheus} as Fig. \ref{plat}. The AprilTag visual system is adopted for the real-time localization of the robots. The control procedures based on the localization results are implemented by MATLAB in a VMWare ESXI virtual machine, which is equipped with an Intel(R) Xeon(R) Gold 5220R CPU, 2.20G Hz processor and 16GB RAM. All experiments are conducted on a 5m $\times$ 3m square platform, and two 17.5cm $\times$ 17.5cm $\times$ 20cm omni-directional mobile cars are used.

\begin{figure*}[ht]
  \centering 
  \setlength{\abovecaptionskip}{0.1cm}
  \subfigcapskip=8pt\subfigure[]
    {
\includegraphics[align=t,width=0.22 \textwidth]{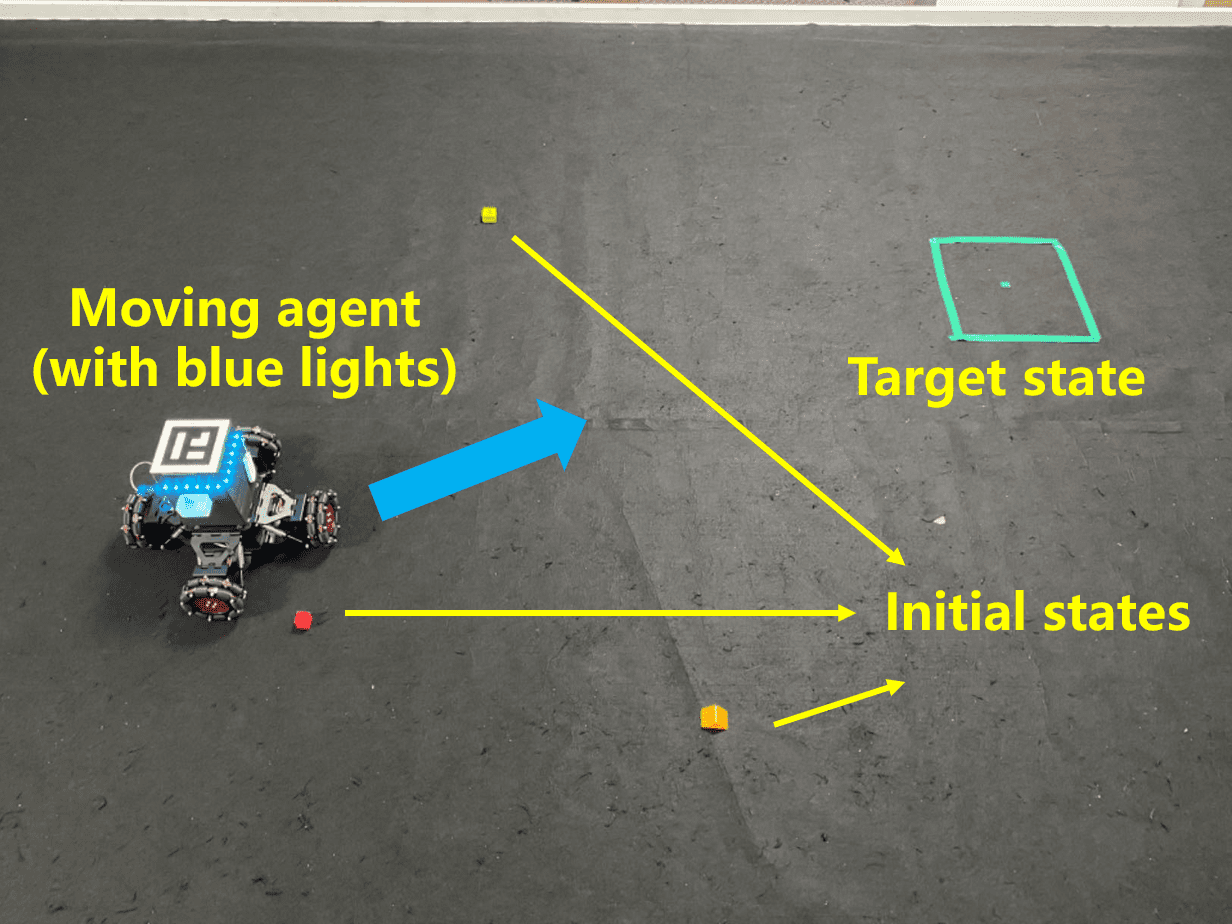} 
    \label{n1}}
    \subfigcapskip=8pt\subfigure[]
    { 
    \label{n2} 
\includegraphics[align=t,width=0.22\textwidth]{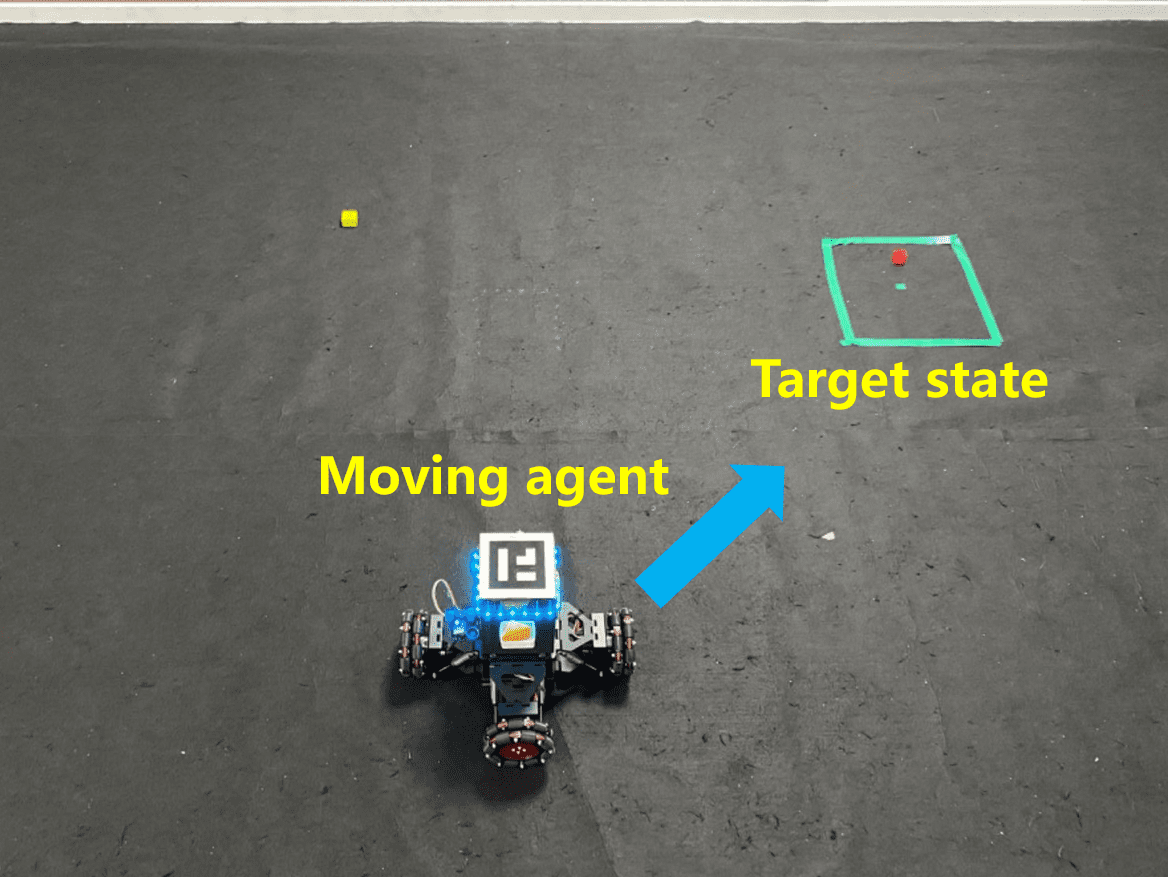} }
    \subfigcapskip=8pt\subfigure[]{ 
    \label{n3} \includegraphics[align=t,width=0.22\textwidth]{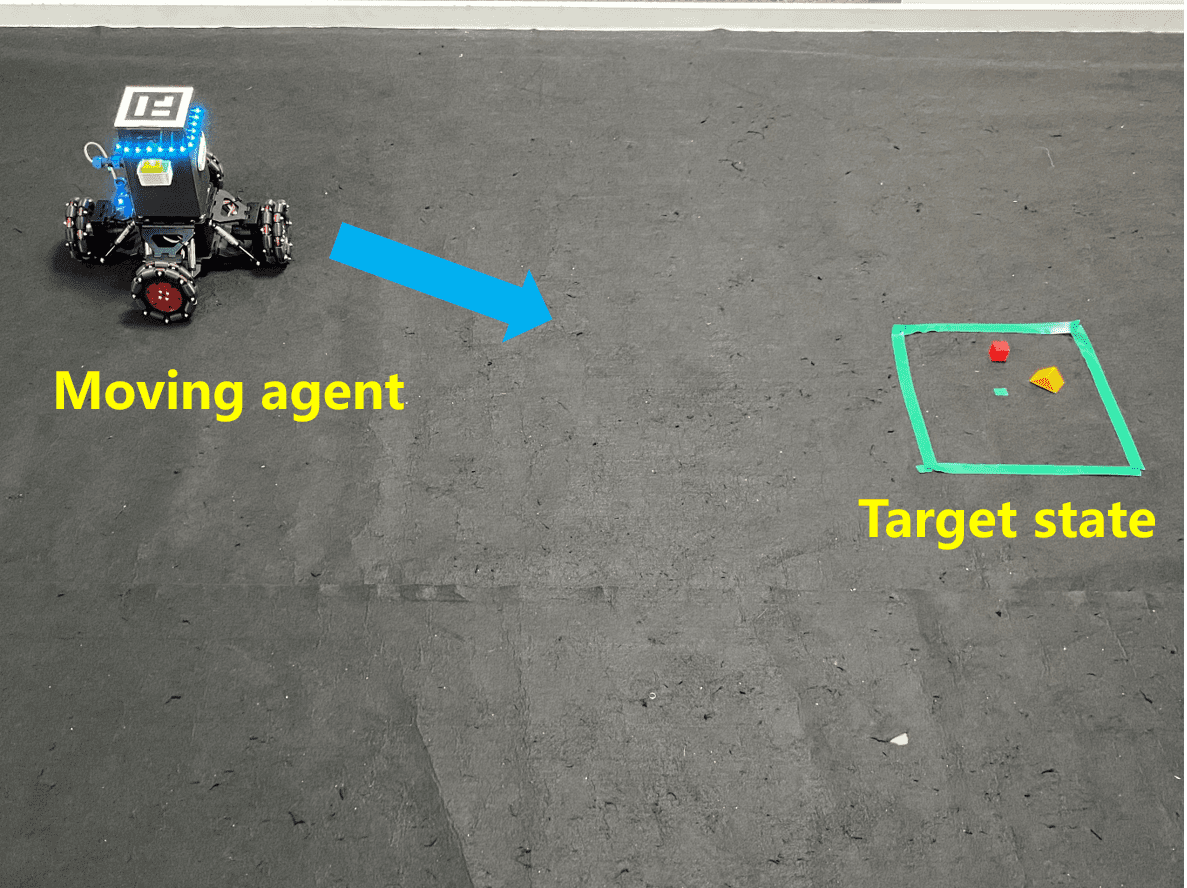} }
    \subfigcapskip=6pt
    \subfigure[]{ 
    \label{n4} 
\includegraphics[align=t,width=0.23\textwidth]{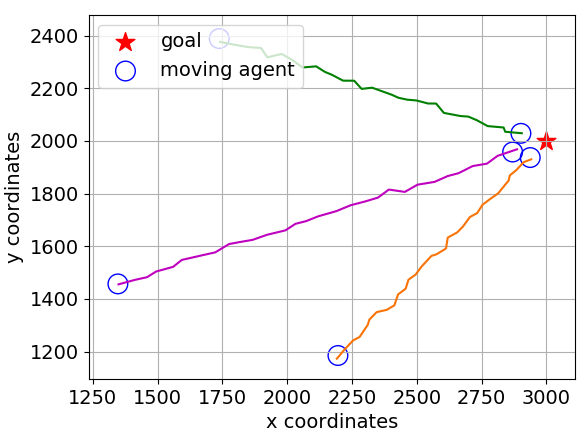} 
    }
\\
  \subfigcapskip=8pt\subfigure[]
    {
    \includegraphics[align=t,width=0.22\textwidth]{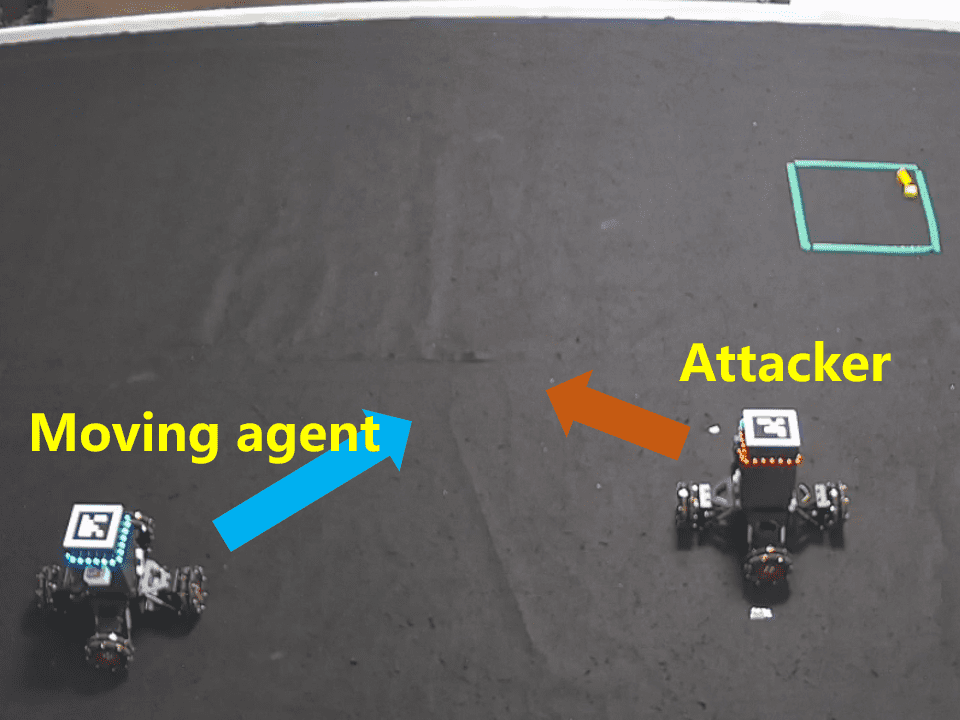} 
    \label{stimulate}
    }
    \subfigcapskip=8pt\subfigure[]
    { 
    \label{speed} \includegraphics[align=t,width=0.22\textwidth]{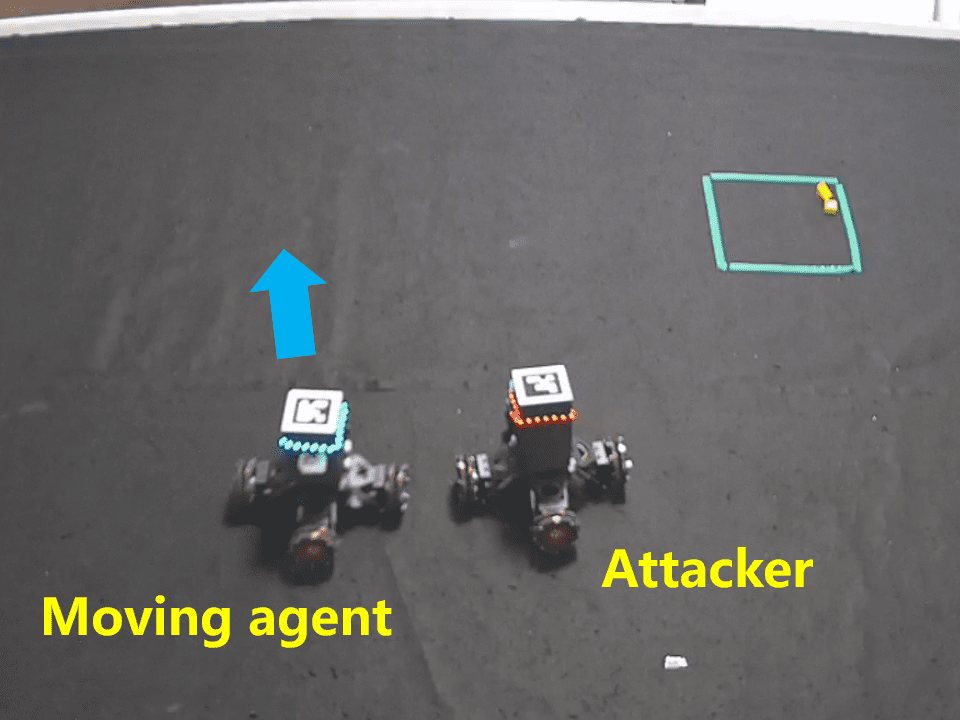} 
    }
    \subfigcapskip=8pt\subfigure[]{ 
    \label{acc} \includegraphics[align=t,width=0.22\textwidth]{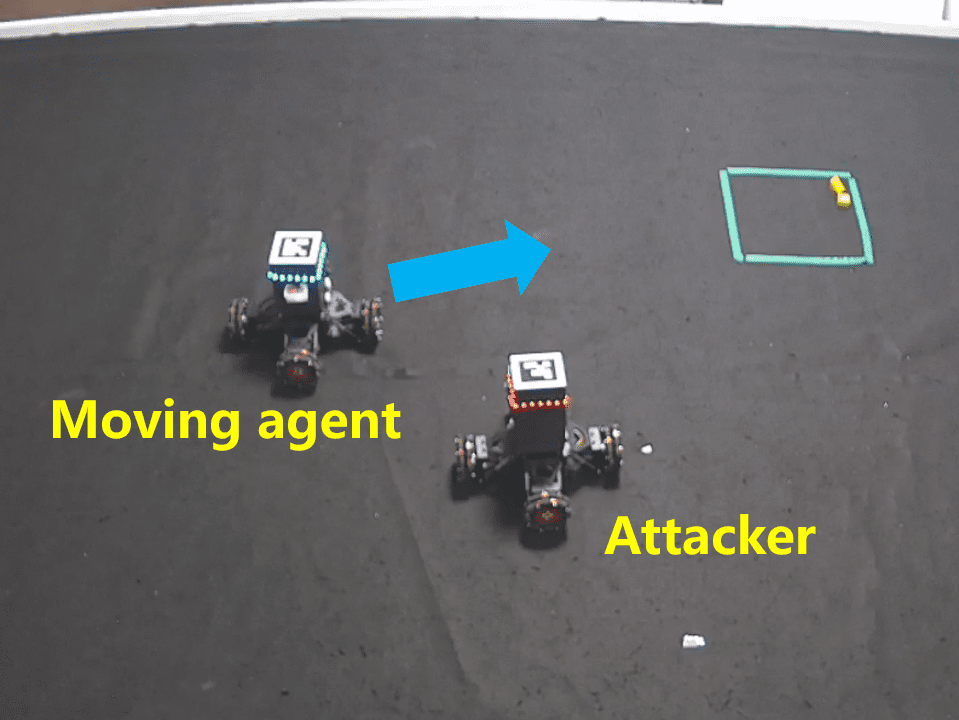} 
    }
    \subfigcapskip=6pt\subfigure[]{ 
    \label{n8} 
\includegraphics[align=t,width=0.23\textwidth]{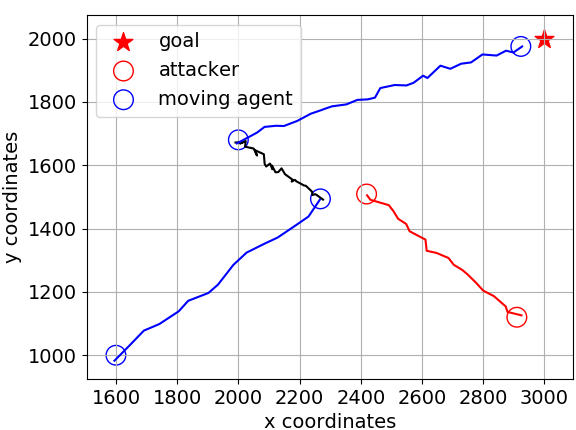} 
    }
  \caption{The illustration of data collection process. The green square is the target state and the moving agent with blue lights drives to the target from different initial position, being observed by the attacker with red lights, which provides multiple unparalleled trajectory observations.}
  \label{collect}
\end{figure*}

\begin{figure}[t]
\setlength{\abovecaptionskip}{0.1cm}
  \centering 
    \subfigure{ 
    \label{pre0} \includegraphics[align=c,width=0.22\textwidth]{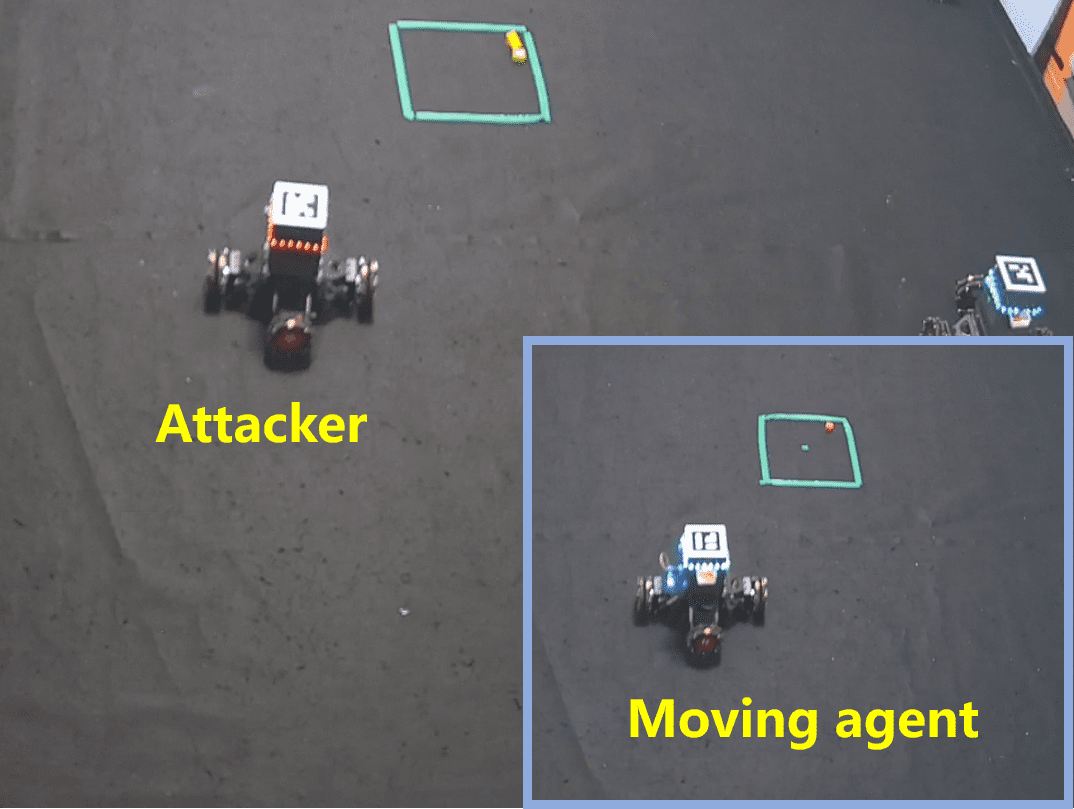} 
    }
    \subfigure{ 
    \label{pre1} \includegraphics[align=c,width=0.22\textwidth]{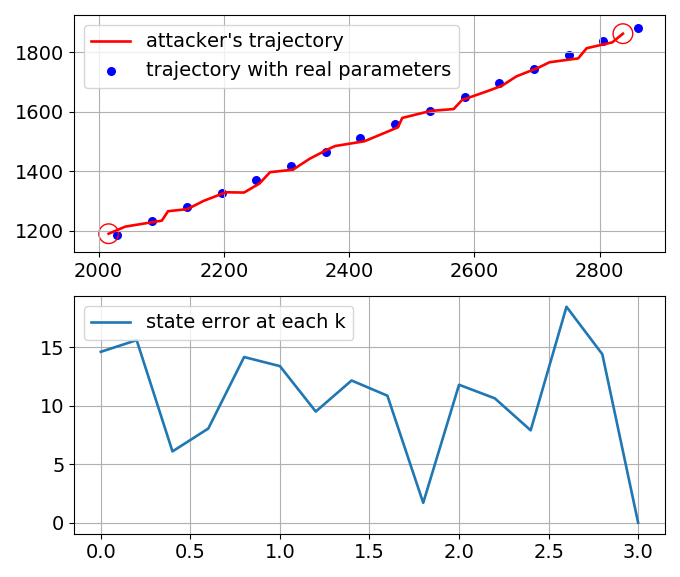} 
    }
    \\
    \subfigure{ 
    \label{pre2} \includegraphics[align=c,width=0.22\textwidth]{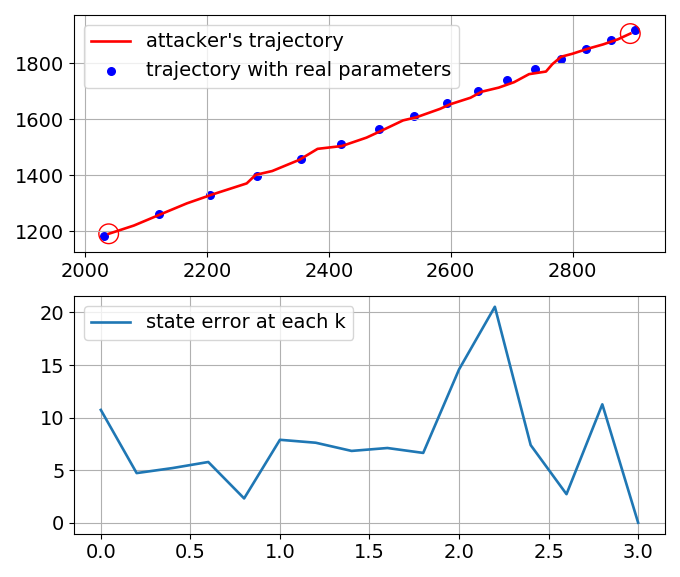} 
    }
    \subfigure{ 
    \label{pre3} \includegraphics[align=c,width=0.22\textwidth]{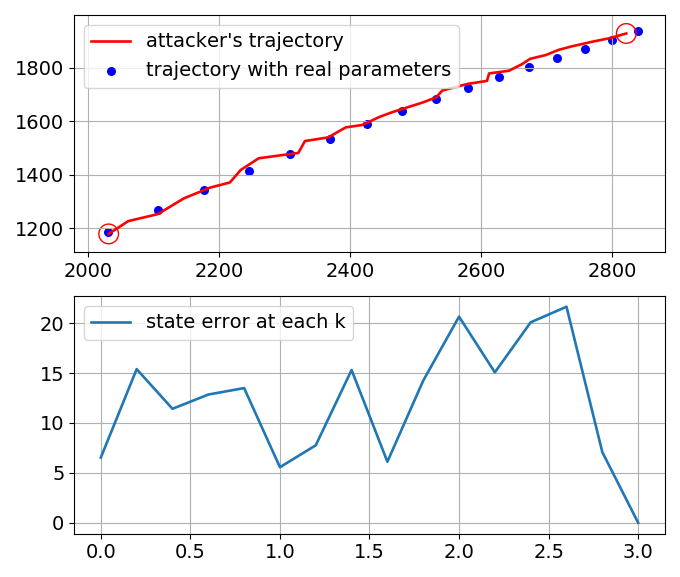} 
    }
  \caption{The attacker learns the control law and imitates trajectories of the moving agent under different motions including uniform linear, variable speed linear and variable speed curve motion.} 
  \label{mimic_pic}
\end{figure}

\begin{figure}[t]
\centering
\includegraphics[width=0.28\textwidth]{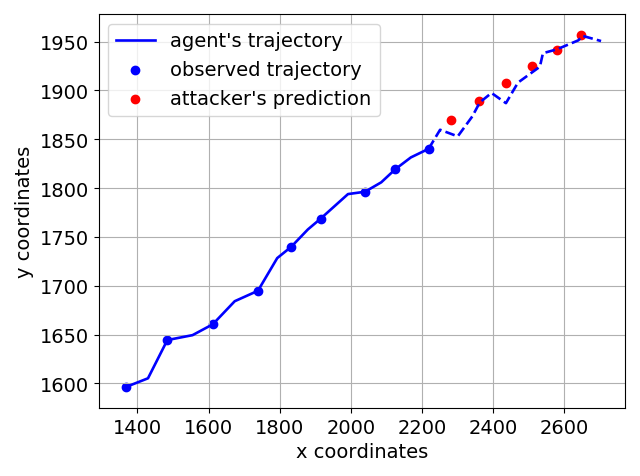}
\caption{Conduct the LQR reconstruction based input prediction algorithm on the self-designed platform. The blue line (and the dash line) is the agent's real trajectory. The blue scatters are observed states sequence and the red are our future state prediction.}
\label{exp_pre}
\end{figure}

We simulate a carrying scenario. A moving agent equipped with the blue light strip transports building blocks from different locations to the green box under LQR optimal control, while the car with red lights acts as an external attacker to observe and record the trajectory of the blue one in order to reconstruct its optimization problem, thereby inferring the future state or imitating its behavior. We have the system state $x_k = (x_k^x \, x_k^y)^T$ and input vector $u_k = (u_k^x \, u_k^y)^T$, where the superscript indicates the horizontal or vertical direction. The control problem of the moving agent is described as $\mathbf{P}_0$ with $A =C= I_2, B=0.2 I_2$ and $N=15, H=5 I_2,Q=0.1 I_2, R=0.5 I_2$. Set $x_T = (3000,2000)^T$ as the target state. 

During the data collection period, two kinds of conditions described in Fig. \ref{data_col} are considered. Notice that in scenario 1, $R_m$ will re-calculate its trajectory after being disturbed. Therefore, we can leverage this mechanism and actively apply external inputs to collect multiple different trajectories. Referring to Fig. \ref{collect}(a)-(d), the moving agent starts from three initial points including $[1430, 1457], [2196,1185], [1738,2389]$ in sequence and drives to the set target, which provides three unparalleled trajectory observations. Another situation is shown as Fig. \ref{collect}(e)-(h) which leverages the attacker as an external stimulate and makes the agent to re-plan from its current state, leading to two different trajectories. Therefore, we take four trajectories as observation data and launch our algorithm.
We have $\hat{x}_T=(2973.3, 1989)^T$ and the estimation results of objective function parameters are
\begin{equation}\nonumber
\begin{aligned}
& \hat{H} = \begin{bmatrix}
    49.0753 &  -2.1655\\
   -2.1655  & 48.6586\\
\end{bmatrix},
\hat{Q} = \begin{bmatrix}
    1.0025 &  -0.0060\\
   -0.0060  &  1.0143\\
\end{bmatrix},\\
&\hat{R} = \begin{bmatrix}
    4.8863 &  -0.2250\\
   -0.2250 &   4.8459\\
\end{bmatrix}.
\end{aligned}
\end{equation}

When the attacker successfully reconstructs the control optimization problem with the estimates of all the parameters, we demonstrate that starting from an arbitrary starting point, the attacker can mimic the trajectories of the agent under different motions and accurately predict the future state with Algorithm \ref{alg2} as shown in Fig. \ref{mimic_pic} and \ref{exp_pre} respectively.
All experiments demonstrate the effectiveness of the proposed reconstruction algorithm. More tests including variable speed and curve motions can be found in the video {https://youtu.be/-rUWd9k3-mg}.

\section{Conclusion}\label{conclusion}\label{conc}
This paper proposes an algorithm for learning the control law of a moving agent under LQR control with optimization problem reconstruction based on observations. Assuming the linear dynamical system and the quadratic objective function form is known, we identify the parameters in the objective function based on inverse optimal control considering of settings and estimate the control horizon with a binary search method. Finally, we reconstruct the control optimization problem of the agent and calculate its solution as the learned control law. We apply our algorithm to predict the future input and conduct extensive simulations and hardware experiments to demonstrate its efficiency. Future work may consider nonlinear modeled agents and identifying more complex objective functions with constraints and obstacles in the environment.

\vspace{5pt}
\begin{appendices}
\section{Proof of Lemma \ref{lem_hr}}\label{pr1}
Apply the problem constraints to the new parameters $H',R'$. We obtain that $(\lambda_i^j)' = \alpha \lambda_i^j$ for all $i,j$ and
\begin{align}\nonumber
(x^j_{i+1})'& = A (x^j_i)' - B(R')^{-1}B^T (\lambda^j_{i+1})'\\ \nonumber & = A (x^j_i)' - BR^{-1} B^T \lambda_{i+1}^j.
\end{align}
Comparing to $x^j_{i+1} = A x^j_i - BR^{-1} B^T \lambda_{i+1}^j$, if we have $(x_0^j)' = x_0^j$, it is easy to realize that $H',R',x_{1:N}^j,(\lambda_{1:N}^j)'$ are also solutions to the problem with the same objective function value of the real parameters $H,R$.

\section{Proof of Theorem \ref{thm_wp}}\label{pr2}
Now we will show that if $A^c_k$ and ${A^c_k}'$ defined in the theorem are equal for all $k$, then we have $R=R'$. We prove this by contradiction. Suppose the two corresponding positive definite matrices $R,R'$ are different, then define
\[
\Delta R = R'-R \neq 0,
\]
where $\Delta R$ is also a symmetric matrix.
Then there are two sets of matrices $P_{0:N}, K_{0:N-1}$ and $P_{0:N}', K_{0:N-1}'$ satisfying the iteration equations:
\begin{align}\label{K}
& K_k =  (R + B^T P_{k+1} B)^{-1} B^T P_{k+1} A,\\
& P_k =K_k^TR K_k + {A^c_k}^T P_{k+1} A^c_k,\label{P}
\end{align}
with $P_{N} = H=I$ respectively. Since we have $A^c_k = {A^c_k}'$ for all $k$, there is
\[
A-B K_k \! =\! A-BK_k' \Leftrightarrow B K_k \!=\! BK_k' \Leftrightarrow B^TB K_k \!=\! B^T B K_k'.
\]
Note that $B$ has full column rank, thus $B^TB$ is invertible. Therefore, we derive $K_k = K_k'$ directly.

According to \eqref{K}, it follows that
\begin{align}\nonumber
(R + B^T P_{k+1} B) K_k &= B^T P_{k+1} A, \\\nonumber
R K_k &= B^T P_{k+1} A^c_k.
\end{align}
Then we also have the above equation for $R'$ and $P_{k+1}'$ which is written as
\begin{align}\nonumber
(R + \Delta R) K_k& = B^T (P_{k+1} + \Delta P_{k+1}) A^c_k,\\ 
\Delta R K_k& = B^T \Delta P_{k+1} A^c_k, \label{delta_r}
\end{align}
where $\Delta P_{k} = P'_k-P_k$.
Similarly, for equation \eqref{P} we have
\begin{align}\nonumber
P_k + \Delta P_k& =K_k^T (R + \Delta R) K_k + {A^c_k}^T (P_{k+1} +\Delta P_{k+1} A^c_k),\\
\Delta P_k& = (K_k^T B^T + {A^c_k}^T) \Delta P_{k+1} A^c_k.\label{delta_p}
\end{align}
Since $P_N = P_N' = I$, then $\Delta P_N = 0$. Combining with \eqref{delta_p}, we have $\Delta P_k = 0$ and $P_k = P_k'$ for all $k$. Thus \eqref{delta_r} converts into
\[
\Delta R K_k = 0, k = 0,1,\dots , N-1.
\]
Note that $I\succ 0, R\succ 0$, then with \eqref{P} we can obtain that $P_k \succ 0$ and is invertible. Therefore, since $P_k, A$ are all invertible matrices, from \eqref{K} we derive $rank(K_k) = rank(B^T) = m$. Thus $K_k$ has full row rank and $\Delta R =0$, which is contradict with the assumption. The proof is done.

\section{Proof of Therorem \ref{thm_r}}\label{pr3}
Denote $z_i = \begin{pmatrix} x_i^T & \lambda_i^T
\end{pmatrix}^T, i = 1, \dots,N$. Then the constraints for each step in Problem \ref{prob_r} is written as
\begin{equation}
\underbrace{\begin{bmatrix}
I & B R^{-1} B^T \\ 0 & A^T
\end{bmatrix}}_E z_{i+1} = 
\underbrace{\begin{bmatrix}
A & 0 \\ 0 & I
\end{bmatrix}}_F z_i.
\end{equation}
Combine all the constraints into a matrix equation as follows
\begin{equation}\label{fzb}
\underbrace{\begin{bmatrix}
\Tilde{E} & & & \Tilde{F} \\
-F & E & & \\
 & \ddots & \ddots & \\
 & & -F & E \\
\end{bmatrix}}_{\mathscr{F}(R)}
\underbrace{\begin{bmatrix}
z_1 \\ z_2 \\ \vdots \\ z_N
\end{bmatrix}}_Z
= \underbrace{\begin{bmatrix}
A \\ 0 \\ \vdots \\ 0
\end{bmatrix}}_{\Tilde{A}}x_0,
\end{equation}
where $\Tilde{E} = \begin{bmatrix}
I & B R^{-1} B^T \\ 0 & 0
\end{bmatrix}$ and $\Tilde{F} = \begin{bmatrix}
0 & 0 \\ -I & I
\end{bmatrix}$. Note that $\mathscr{F}(R)$ is an invertible matrix. Then we have
\begin{equation}\nonumber
\sum_{i = 1}^N \Vert y_i - x^*_i \Vert^2 = \Vert Y - G_X Z \Vert^2 = \Vert Y - G_X \mathscr{F}(R)^{-1} \Tilde{A} x_0 \Vert^2,
\end{equation}
where $G_X = I_{N-1} \otimes [I_n, 0_n]$. 
Then the subsequent proof is similar to the proof of Theorem 4.1 in \cite{zhang2019inverse}. Just replace $\mathscr{F}(Q)$ in \cite{zhang2019inverse} with $\mathscr{F}(R)$ in \eqref{fzb} and we have $\hat{R} \xrightarrow{\mathrm{P}} R$ as $M \xrightarrow{} \infty$, where $\hat{R}$ is the solution obtained by Problem \ref{prob_r} and $R$ is the true parameter in the forward problem.

\section{Proof of Theorem \ref{hqr_iden}}\label{hqr_iden_pr}
This proof will show that the matrix pairs $H,Q,R$ and $H',Q',R'$ obtained through the iteration equations \eqref{dis_k} with same sequence $K_{0:N-1}$ satisfy a scalar multiple relationship under some linearly independent conditions.

From \eqref{dis_k} we have
\begin{equation}\nonumber
\begin{aligned}
& A_k^c= A-BK_k = A- B (R + B^T P_{k+1} B)^{-1} B^T P_{k+1} A \\
& = (I+BR^{-1}B^T P_{k+1})^{-1} A,
\end{aligned}
\end{equation}
where we used the Sherman–Morrison formula.
Since $K_{0:N-1}$ and $K'_{0:N-1}$ are the same, $A_k^c$'s corresponding to $K_k$ and $K'_k$ are the same and we have
\begin{equation}\label{brp}
R^{-1}B^T P_{k+1} = R'^{-1}B^T P'_{k+1}
\end{equation}
for all $k$, where we used that $A$ is invertible and $B$ has full column rank. For $P_k$, we have
\begin{equation}\nonumber
\begin{aligned}
& P_{k-1} = A^TP_{k}A - A^T P_{k} B(R+ B^T P_{k}B)^{-1} B^TP_{k}A+Q \\
& =A^T P_k A_{k-1}^c +Q.
\end{aligned}
\end{equation}
Starting from $k=N$ with $P_N=H,P'_N=H'$, we substitute $P_k$ into \eqref{brp} and obtain the following $N$ equations.
\begin{equation}\nonumber
\begin{aligned}
& R^{-1}B^T H = R'^{-1}B^T H', \\
& R^{-1}B^T (Q+A^T H A_{N-1}^c)=R'^{-1}B^T (Q'+A^T H' A_{N-1}^c),\\
& R^{-1}B^T (Q+A^T Q A_{N-2}^c + (A^T)^2 H A_{N-1}^c A_{N-2}^c)=\\
& R'^{-1}B^T (Q'+A^T Q' A_{N-2}^c + (A^T)^2 H' A_{N-1}^c A_{N-2}^c),\\
& \cdots 
\end{aligned}
\end{equation}
In particular, for $k=i$, we have $R^{-1}B^T P_i = R'^{-1}B^T P'_i$, where
\[
P_i = Q+ \Sigma_{j=1}^{N-i} (A^T)^j Q \Pi_{r=i}^{i+j-1} A_r^c + (A^T)^{N-i} H \Pi_{r=i}^{N-1} A_{r}^c.
\]
We subtract the equation for $k=i+1$ from the equation for $k=i$ and write the difference in the following matrix form:
\begin{equation}\label{bqr_matrix}
R^{-1} \widetilde{\mathcal{BA}}_i \widetilde{QH}_i \widetilde{\mathcal{A}^c}_i = R'^{-1} \widetilde{\mathcal{BA}}_i \widetilde{Q'H'}_i \widetilde{\mathcal{A}^c}_i,
\end{equation}
where 
\begin{equation}\nonumber
\begin{aligned}
& \widetilde{\mathcal{BA}}_i = B^T [A^T, \cdots, (A^{T})^{N-i-1}, (A^{T})^{N-i-1}, (A^{T})^{N-i}],\\
&\widetilde{\mathcal{A}^c}_i \!=\! [({A^c_i}\!-\!{A^c_{i+1}})^T, \!\cdots\!, (\Pi_{r=0}^{p-1}{A^c_{i+r}}\!-\!\Pi_{r=1}^{p}{A^c_{i+r}})^T,\cdots,\\
& ~~~~~~~ (\Pi_{r=0}^{N-1}{A^c_{i+r}})^T, -(\Pi_{r=1}^{N}{A^c_{i+r}})^T, (\Pi_{r=0}^{N}{A^c_{i+r}})^T]^T,\\
& \widetilde{QH}_i = \begin{bmatrix}
I_{N-i-1} \otimes Q& 0 \\ 0& I_{2} \otimes H
\end{bmatrix}.
\end{aligned}
\end{equation}
$\widetilde{Q'H'}_i$ can be obtained by replacing $Q,H$ in $\widetilde{QH}_i$ with $Q',H'$. Index $p$ in the block matrix $\widetilde{\mathcal{A}^c}_i$ refers to the $p$-th block, i.e, for the first block, $p=1$ and it can directly write as $({A^c_i}\!-\!{A^c_{i+1}})^T$.
Take the trace of both sides of equation \eqref{bqr_matrix} and we have
\begin{equation}\label{trace}
\begin{aligned}
\mathrm{tr}(R^{-1} \widetilde{\mathcal{BA}}_i \widetilde{QH}_i \widetilde{\mathcal{A}^c}_i) = 
vec(R^{-1})^T \underbrace{(\widetilde{\mathcal{A}^c}_i^T \otimes \widetilde{\mathcal{BA}}_i)}_{\mathcal{E}_i} vec(\widetilde{QH}_i).
\end{aligned}
\end{equation}

Notice that there are many zero elements and identical matrices in $\widetilde{QH}_i$ and $R^{-1}$. Therefore, the trace equation \eqref{trace} can be simplified through the following three steps:\\
i) Let the set of indices of all non-zero elements in $vec(\widetilde{QH}_i)$ be $\mathcal{C}_1$. Collect all the columns of $\mathcal{E}_i$ with indices in $\mathcal{C}_1$ and form ${\mathcal{P}^1_i}(\mathcal{E}_i)= \mathcal{E}_i(:,\mathcal{C}_1)$. Then we have
\[
vec(R^{-1})^T {\mathcal{P}^1_i}(\mathcal{E}_i) [\textbf{1}_{N-i-1}^T \otimes vec(Q)^T, vec(H)^T,vec(H)^T]^T.
\]
ii) Since $vec(Q)$ and $vec(H)$ appear multiple times, we define $\mathcal{P}^2_i(\mathcal{E}_i)$, where 
\begin{equation}\nonumber
\begin{aligned}
& \mathcal{P}^2_i(\mathcal{E}_i)(:,1:n)=\Sigma_{t=1}^{N-i-1}{\mathcal{P}^1_i}(\mathcal{E}_i)(:,1+m(t-1):mt),\\
& \mathcal{P}^2_i(\mathcal{E}_i)(:,n+1:2n)=\Sigma_{t=N-i}^{N-i+1}{\mathcal{P}^1_i}(\mathcal{E}_i)(:,1+m(t-1):mt).
\end{aligned}
\end{equation}
iii) Since $H,Q \in \mathbb{R}^n$ and $R \in \mathbb{R}^m$ are symmetric, we only need the upper triangle part $vec(\overline{H}), vec(\overline{Q}), vec(\overline{R^{-1}})$ respectively to uniquely determine them. We can further simplify \eqref{trace} as
\begin{equation}\label{pie}
\begin{aligned}
& vec(\overline{R^{-1}})^T \mathcal{P}_i(\mathcal{E}_i) \begin{bmatrix}
vec(\overline{Q}) \\vec(\overline{H})
\end{bmatrix}\\
& = (\begin{bmatrix}
vec(\overline{Q})^T, vec(\overline{H})^T
\end{bmatrix} \otimes vec(\overline{R^{-1}})^T) vec(\mathcal{P}_i(\mathcal{E}_i)) \\
& = (\begin{bmatrix}
vec(\overline{Q'})^T, vec(\overline{H'})^T
\end{bmatrix} \otimes vec(\overline{R'^{-1}})^T) vec(\mathcal{P}_i(\mathcal{E}_i)),
\end{aligned}
\end{equation}
where $\mathcal{P}_i(\mathcal{E}_i)$ is a $\frac{m(m+1)}{2} \times n(n+1)$ matrix. 

If there exist at least $\frac{mn(n+1)(m+1)}{2}$ linearly independent vectors 
$vec(\mathcal{P}_i(\mathcal{E}_i))$ in the horizon $k=0,\dots, N-1$, then we can combine them as a full row rank matrix and obtain
\begin{equation}\nonumber
\begin{aligned}
& vec(\overline{Q'})=\alpha \cdot vec(\overline{Q}), vec(\overline{H'}) =\alpha \cdot vec(\overline{H}),\\ 
&vec(\overline{R'^{-1}}) =\frac{1}{\alpha} \cdot vec(\overline{R^{-1}}).
\end{aligned}
\end{equation} 
Thus, we have ${H'} = \alpha H, {Q'} = \alpha Q, {R'}=\alpha R$.
Note that if the matrix $H,Q,R$ are all diagonal matrices, which is a common setting in practical scenarios, we only need $2nm$ linearly independent $vec(\mathcal{P}_i(\mathcal{E}_i))$.

\section{Proof of Lemma \ref{lem_rq_eq}}\label{pr_lem_rq}
In this proof we will show how to obtain the equation \eqref{rq_eq} from the iterations of $P_k, K_k$. Transform the formula \eqref{dis_k}:
\begin{equation}\nonumber
\begin{aligned}
\eqref{dis_k} & \Leftrightarrow (R + B^T P_{k+1} B) K_k = B^T P_{k+1} A, \\
& \Leftrightarrow R K_k = B^T P_{k+1} (A - B K_k) = B^T P_{k+1} A^c_k.
\end{aligned}
\end{equation}
Since $P_N = H$, for $k=N-1$ we have
\[
R K_{N-1} = B^T H A^c_{N-1} \Leftrightarrow a_1 R b_1 = c_1 H d_1
\]
where $a_1 = I_2, b_1 = K_{N-1}, c_1 = B^T$ and $d_1 = A^c_{N-1}$. For $k=N-2$ there is
\[
R K_{N-2} = B^T P_{N-1} A^c_{N-2}
\]
and substitute $P_{N-1}$ with \eqref{dis_p}
\begin{equation}\nonumber
\begin{aligned}
& R K_{N-2} = B^T (K_{N-1}^T R K_{N-1} + {A^c_{N-1}}^T H A^c_{N-1}+Q) A^c_{N-2}\\
& \Leftrightarrow \underbrace{\begin{bmatrix}I_2 & - B^T K_{N-1}^T\end{bmatrix}}_{a_2} (I_2 \otimes R) \underbrace{\begin{bmatrix}K_{N-2} \\ K_{N-1} A^c_{N-2}\end{bmatrix}}_{b_2} = \\
& ~~~~ \underbrace{\begin{bmatrix}B^T {A^c_{N-1}}^T & B^T\end{bmatrix}}_{c_2} \begin{bmatrix}H & \\ & Q\end{bmatrix} \underbrace{\begin{bmatrix}A^c_{N-1} A^c_{N-2} \\ A^c_{N-2}\end{bmatrix}}_{d_2}.
\end{aligned}
\end{equation}
For $k = N-3$ we have
\begin{equation}\nonumber
\begin{aligned}
& R K_{N-3} = B^T P_{N-2} A^c_{N-3} \\
&= B^T (K_{N-2}^T R K_{N-2} + {A^c_{N-2}}^T P_{N-1} A^c_{N-2}+Q) A^c_{N-3}\\
& = B^T (K_{N-2}^T R K_{N-2} + {A^c_{N-2}}^T (K_{N-1}^T R K_{N-1} \\
& ~~~+ {A^c_{N-1}}^T H A^c_{N-1}+Q) A^c_{N-2}+Q) A^c_{N-3},\\
\end{aligned}
\end{equation}
which can also be transformed into the form 
\[a_3 (I_3 \otimes R) b_3 = c_3 \begin{bmatrix}H &\\& I_2 \otimes Q\end{bmatrix} d_3.\]
Therefore, according to the above derivation we conclude the equation \eqref{rq_eq} for $i=1,\dots,N$, where the parameters $a_i,b_i,c_i,d_i$ are calculated by \eqref{abcd}.

\section{Proof of Theorem \ref{thm_sens}}\label{prf_sens}

\textbf{\emph{a})} We first studies the convergence of $\{P_k\}$ sequence. 
The following lemma provides a convergent property of algebraic Riccati equation.
\begin{lemma}\label{p_conv}
There exist a constant $\gamma >1$ and two non-negative constants $c_1 = c_1(\gamma, P_N, P_{N-1})$ and $c_2 = c_2(\gamma, P_N, P_{N-1})$ such that for any fixed finite control horizon $N$ the following inequality holds for all $k=0,1,\dots,N-1$:
\begin{equation}
\underline{\beta}_k P_{k+1} \leqslant P_{k} \leqslant \overline{\beta}_k P_{k+1},
\end{equation}
where
\begin{equation}\nonumber
\underline{\beta}_k := \frac{\gamma^k (\gamma-1)}{\gamma^k(\gamma-1)+c_1}, \, \overline{\beta}_k := \frac{\gamma^k (\gamma-1) + c_2}{\gamma^k(\gamma-1)}.
\end{equation}
\end{lemma}
\begin{proof}
See the proof of Proposition $7$ in \cite{cai2017convergent}.
\end{proof}
\noindent Lemma \ref{p_conv} reveals both monotonic and convergent properties of the $\{P_k\}$ sequence. 
Denote a positive definite matrix $\Phi=P_N^{-1/2} P_{N-1} P_N^{-1/2}$ and let
\begin{equation}
\begin{array}{ll}
& c_1 = \max \{0, (\frac{1}{\lambda_{\min}(\Phi)}-1)(\gamma -1)\}, \\
& c_2 = \max \{0, (\lambda_{\max}(\Phi)-1)(\gamma -1)\}.
\end{array}
\end{equation}
The corresponding parameter $\gamma := 1/(1-\sigma)$ is calculated as formula (16) in \cite{cai2017convergent}. However, one basic assumptions of Lemma \ref{p_conv} is $Q>0$, while our problem may have $Q=0$ (in Sec. \ref{ioc}-B). If $Q=0$, their calculation of parameter $\gamma$ will fail and we provide the following selection method of $\gamma$ instead.

Since the sequence $\{P_k\}$ is monotonic, we will discuss two kinds of conditions: increasing and decreasing.
Firstly, we have
\begin{align}\nonumber
& P_k - (A \!-\! B K_k )^T P_{k+1} (A - B K_k) =K_k^TR K_k\!= \!A^T\! \\ \nonumber
& P_{k+1}^T \!B (R\! +\! B^T \!P_{k+1} B)^{-T} \! R ({R}\! +\! B^T\! P_{k+1} B)^{-1}\! B^T\! P_{k+1} A.
\end{align}

i) If the sequence $\{P_k\}$ is monotonically decreasing, there is $P_N \geqslant P_k \geqslant P^*$ for all $k$. Thus we have
\begin{align}\nonumber
& P_k - (A \!-\! B K_k )^T P_{k+1} (A - B K_k) \\ \nonumber
& \geqslant \!A^T\! {P^*}^T \!B (R\! +\! B^T \!P_{N} B)^{-T} \! R ({R}\! +\! B^T\! P_{N} B)^{-1}\! B^T\! P^* A \\ \nonumber
& := Q_a >0.
\end{align}

ii) If the sequence $\{P_k\}$ is monotonically increasing, there is $P_N \leqslant P_k \leqslant P^*$ for all $k$. Thus we have
\begin{align}\nonumber
& P_k - (A \!-\! B K_k )^T P_{k+1} (A - B K_k) \\ \nonumber
& \geqslant \!A^T\! {P_{N}}^T \!B (R\! +\! B^T \!P^* B)^{-T} \! R ({R}\! +\! B^T\! P^* B)^{-1}\! B^T\! P_{N} A\\ \nonumber
& := Q_b >0.
\end{align}
Therefore, let
\[
\Gamma := \frac{1}{\kappa} \cdot Q_a(\mathrm{or} \,\,Q_b), \, \sigma = \lambda_{\min}(\Gamma).
\]
Select $\gamma := 1/(1-\sigma)$ and we have $\gamma>1$.

\textbf{\emph{b})} We will then figure out the convergence property of $\{K_k\}$ with Lemma \ref{p_conv}. Similarly, we discuss the following two monotonic situations separately:

i) If $P_N \leqslant P_{N-1}$ holds, we have $c_1 = 0, c_2 \geqslant 0$ and $P_k \leqslant P_{k-1} \leqslant \overline{\beta}_k P_{k+1}$ for all $k$, which yields that
\begin{align}\nonumber
P_{k}-P_{k+1} &\leqslant \overline{\beta}_k P_{k+1} -  P_{k+1} \\ \nonumber 
& \leqslant (\overline{\beta}_k-1) P_{k+1}\leqslant (\overline{\beta}_k-1) P_{N}.
\end{align}
Then $\Vert P_k-P_{k+1} \Vert \leqslant (\overline{\beta}_k-1) \Vert P_N \Vert$. There always exists $N_a \in \mathbb{N}_+$ and $\epsilon > 0$, for any $N>N_a$, we obtain $\Vert P_{1}-P_0 \Vert \leqslant \epsilon$.
\begin{align}\nonumber
&\Vert K_{0}-K_1 \Vert \\ \nonumber
& = \Vert ({R} + B^T P_1 B)^{-1} B^T P_1 A - ({R} + B^T P_2 B)^{-1} B^T P_2 A \Vert \\ \nonumber
& \leqslant \Vert ({R} + B^T P_2 B)^{-1} B^T (P_1-P_2) A \Vert \\ \nonumber
& \leqslant \Vert ({R} + B^T P_N B)^{-1} B^T\Vert \cdot \Vert P_1-P_2 \Vert \cdot \Vert A \Vert =: \eta_a
\end{align}
    
ii) If $P_N \geqslant P_{N-1}$ holds, we have $c_1\geqslant 0, c_2 = 0$ and $\underline{\beta}_k P_{k+1} \leqslant P_k \leqslant P_{k+1}$ for all $k$, which yields that
\begin{align}\nonumber
P_{k+1}-P_k &\leqslant P_{k+1} - \underline{\beta}_k P_{k+1} \\ \nonumber 
& \leqslant (1-\underline{\beta}_k) P_{k+1}\leqslant (1-\underline{\beta}_k) P_{N}.
\end{align}
Then $\Vert P_{k+1}-P_k \Vert \leqslant (1-\underline{\beta}_k) \Vert P_N \Vert$. There always exists $N_b \in \mathbb{N}_+$ and $\epsilon > 0$, for any $N>N_b$, we obtain $\Vert P_{1}-P_0 \Vert \leqslant \epsilon$.
\begin{align}\nonumber
&\Vert K_{1}-K_0 \Vert \\ \nonumber
& = \Vert ({R} + B^T P_2 B)^{-1} B^T P_2 A - ({R} + B^T P_1 B)^{-1} B^T P_1 A \Vert \\ \nonumber
& \leqslant \Vert ({R} + B^T P_1 B)^{-1} B^T (P_2-P_1) A \Vert \\ \nonumber
& \leqslant \Vert ({R} + B^T P^* B)^{-1} B^T\Vert \cdot \Vert P_2-P_1 \Vert \cdot \Vert A \Vert =: \eta_b
\end{align}

\textbf{\emph{c})} Based on previous discussion, we provide the proof of Theorem \ref{thm_sens}. When $N> \overline{N} = N_a(\mathrm{or}\, N_b)$, there is
\begin{equation}\nonumber
\begin{aligned}
\Vert \mu_0^{(N+\delta N)} - \mu_0^{(N)}\Vert &= \Vert K_0^{(N+\delta N)} - K_0^{(N)}\Vert \cdot \Vert x_0 \Vert \\
& =  \Vert K_0^{(N+\delta N)} - K_{\delta N}^{(N+\delta N)}\Vert \cdot \Vert x_0 \Vert \\
& \leqslant \delta N \Vert K_0^{(N+\delta N)} - K_{1}^{(N+\delta N)}\Vert \cdot \Vert x_0 \Vert \\
& \leqslant \delta N \cdot \eta_a(\mathrm{or}\, \eta_b) \cdot \Vert x_0 \Vert =: \eta.
\end{aligned}
\end{equation}
The proof is done.

\end{appendices}

\bibliographystyle{IEEEtran}
\bibliography{reference}

\begin{thebibliography}{10}
\providecommand{\url}[1]{#1}
\csname url@samestyle\endcsname
\providecommand{\newblock}{\relax}
\providecommand{\bibinfo}[2]{#2}
\providecommand{\BIBentrySTDinterwordspacing}{\spaceskip=0pt\relax}
\providecommand{\BIBentryALTinterwordstretchfactor}{4}
\providecommand{\BIBentryALTinterwordspacing}{\spaceskip=\fontdimen2\font plus
\BIBentryALTinterwordstretchfactor\fontdimen3\font minus
  \fontdimen4\font\relax}
\providecommand{\BIBforeignlanguage}[2]{{%
\expandafter\ifx\csname l@#1\endcsname\relax
\typeout{** WARNING: IEEEtran.bst: No hyphenation pattern has been}%
\typeout{** loaded for the language `#1'. Using the pattern for}%
\typeout{** the default language instead.}%
\else
\language=\csname l@#1\endcsname
\fi
#2}}
\providecommand{\BIBdecl}{\relax}
\BIBdecl

\bibitem{qu2023control}
C.~Qu, J.~He, X.~Duan, and S.~Wu, ``Control input inference of mobile agents
  under unknown objective,'' in \emph{IFAC 2023 World Congress}, to be
  published, 2023.

\bibitem{cao2012mobile}
J.~Cao and S.~K. Das, ``Mobile agents and applications in networking and
  distributed computing,'' \emph{Mobile Agents in Networking and Distributed
  Computing}, pp. 1--16, 2012.

\bibitem{tzafestas2018mobile}
S.~G. Tzafestas, ``Mobile robot control and navigation: A global overview,''
  \emph{Journal of Intelligent \& Robotic Systems}, vol.~91, pp. 35--58, 2018.

\bibitem{zhou2023robust}
L.~Zhou and V.~Kumar, ``Robust multi-robot active target tracking against
  sensing and communication attacks,'' \emph{IEEE Transactions on Robotics},
  vol.~39, no.~3, pp. 1768--1780, 2023.

\bibitem{pasqualetti2013attack}
F.~Pasqualetti, F.~D{\"o}rfler, and F.~Bullo, ``Attack detection and
  identification in cyber-physical systems,'' \emph{IEEE Transactions on
  Automatic Control}, vol.~58, no.~11, pp. 2715--2729, 2013.

\bibitem{wu2023secure}
C.~Wu, W.~Yao, W.~Luo, W.~Pan, G.~Sun, H.~Xie, and L.~Wu, ``A secure robot
  learning framework for cyber attack scheduling and countermeasure,''
  \emph{IEEE Transactions on Robotics}, 2023.

\bibitem{ravichandar2020recent}
H.~Ravichandar, A.~S. Polydoros, S.~Chernova, and A.~Billard, ``Recent advances
  in robot learning from demonstration,'' \emph{Annual Review of Control,
  Robotics, and Autonomous Systems}, vol.~3, pp. 297--330, 2020.

\bibitem{jin2022learning}
W.~Jin, T.~D. Murphey, D.~Kuli{\'c}, N.~Ezer, and S.~Mou, ``Learning from
  sparse demonstrations,'' \emph{IEEE Transactions on Robotics}, vol.~39,
  no.~1, pp. 645--664, 2022.

\bibitem{kuderer2015learning}
M.~Kuderer, S.~Gulati, and W.~Burgard, ``Learning driving styles for autonomous
  vehicles from demonstration,'' in \emph{2015 IEEE International Conference on
  Robotics and Automation (ICRA)}.\hskip 1em plus 0.5em minus 0.4em\relax IEEE,
  2015, pp. 2641--2646.

\bibitem{kent2016construction}
D.~Kent, M.~Behrooz, and S.~Chernova, ``Construction of a 3d object recognition
  and manipulation database from grasp demonstrations,'' \emph{Autonomous
  Robots}, vol.~40, pp. 175--192, 2016.

\bibitem{maeda2017probabilistic}
G.~J. Maeda, G.~Neumann, M.~Ewerton, R.~Lioutikov, O.~Kroemer, and J.~Peters,
  ``Probabilistic movement primitives for coordination of multiple human--robot
  collaborative tasks,'' \emph{Autonomous Robots}, vol.~41, pp. 593--612, 2017.

\bibitem{rahmatizadeh2018vision}
R.~Rahmatizadeh, P.~Abolghasemi, L.~B{\"o}l{\"o}ni, and S.~Levine,
  ``Vision-based multi-task manipulation for inexpensive robots using
  end-to-end learning from demonstration,'' in \emph{2018 IEEE International
  Conference on Robotics and Automation (ICRA)}.\hskip 1em plus 0.5em minus
  0.4em\relax IEEE, 2018, pp. 3758--3765.

\bibitem{torabi2018behavioral}
F.~Torabi, G.~Warnell, and P.~Stone, ``Behavioral cloning from observation,''
  \emph{arXiv preprint arXiv:1805.01954}, 2018.

\bibitem{bemporad2002explicit}
A.~Bemporad, M.~Morari, V.~Dua, and E.~N. Pistikopoulos, ``The explicit linear
  quadratic regulator for constrained systems,'' \emph{Automatica}, vol.~38,
  no.~1, pp. 3--20, 2002.

\bibitem{ab2020inverse}
N.~Ab~Azar, A.~Shahmansoorian, and M.~Davoudi, ``From inverse optimal control
  to inverse reinforcement learning: A historical review,'' \emph{Annual
  Reviews in Control}, vol.~50, pp. 119--138, 2020.

\bibitem{jin2019inverse}
W.~Jin, D.~Kuli{\'c}, J.~F.-S. Lin, S.~Mou, and S.~Hirche, ``Inverse optimal
  control for multiphase cost functions,'' \emph{IEEE Transactions on
  Robotics}, vol.~35, no.~6, pp. 1387--1398, 2019.

\bibitem{anderson2007optimal}
B.~D. Anderson and J.~B. Moore, \emph{Optimal control: Linear quadratic
  methods}.\hskip 1em plus 0.5em minus 0.4em\relax Courier Corporation, 2007.

\bibitem{priess2014solutions}
M.~C. Priess, R.~Conway, J.~Choi, J.~M. Popovich, and C.~Radcliffe, ``Solutions
  to the inverse lqr problem with application to biological systems analysis,''
  \emph{IEEE Transactions on Control Systems Technology}, vol.~23, no.~2, pp.
  770--777, 2014.

\bibitem{pauwels2016linear}
E.~Pauwels, D.~Henrion, and J.-B. Lasserre, ``Linear conic optimization for
  inverse optimal control,'' \emph{SIAM Journal on Control and Optimization},
  vol.~54, no.~3, pp. 1798--1825, 2016.

\bibitem{li2020continuous}
Y.~Li, Y.~Yao, and X.~Hu, ``Continuous-time inverse quadratic optimal control
  problem,'' \emph{Automatica}, vol. 117, p. 108977, 2020.

\bibitem{yu2021system}
C.~Yu, Y.~Li, H.~Fang, and J.~Chen, ``System identification approach for
  inverse optimal control of finite-horizon linear quadratic regulators,''
  \emph{Automatica}, vol. 129, p. 109636, 2021.

\bibitem{zhang2019inverse}
H.~Zhang, J.~Umenberger, and X.~Hu, ``Inverse optimal control for discrete-time
  finite-horizon linear quadratic regulators,'' \emph{Automatica}, vol. 110, p.
  108593, 2019.

\bibitem{primbs2000feasibility}
J.~A. Primbs and V.~Nevisti{\'c}, ``Feasibility and stability of constrained
  finite receding horizon control,'' \emph{Automatica}, vol.~36, no.~7, pp.
  965--971, 2000.

\bibitem{bohn2021reinforcement}
E.~B{\o}hn, S.~Gros, S.~Moe, and T.~A. Johansen, ``Reinforcement learning of
  the prediction horizon in model predictive control,''
  \emph{IFAC-PapersOnLine}, vol.~54, no.~6, pp. 314--320, 2021.

\bibitem{sun2019robust}
Z.~Sun, L.~Dai, K.~Liu, D.~V. Dimarogonas, and Y.~Xia, ``Robust self-triggered
  mpc with adaptive prediction horizon for perturbed nonlinear systems,''
  \emph{IEEE Transactions on Automatic Control}, vol.~64, no.~11, pp.
  4780--4787, 2019.

\bibitem{chen2016tracking}
J.~Chen, T.~Liu, and S.~Shen, ``Tracking a moving target in cluttered
  environments using a quadrotor,'' in \emph{2016 IEEE/RSJ International
  Conference on Intelligent Robots and Systems (IROS)}.\hskip 1em plus 0.5em
  minus 0.4em\relax IEEE, 2016, pp. 446--453.

\bibitem{gao2018online}
F.~Gao, W.~Wu, Y.~Lin, and S.~Shen, ``Online safe trajectory generation for
  quadrotors using fast marching method and bernstein basis polynomial,'' in
  \emph{2018 IEEE International Conference on Robotics and Automation
  (ICRA)}.\hskip 1em plus 0.5em minus 0.4em\relax IEEE, 2018, pp. 344--351.

\bibitem{altche2017lstm}
F.~Altch{\'e} and A.~de~La~Fortelle, ``An lstm network for highway trajectory
  prediction,'' in \emph{2017 IEEE 20th International Conference on Intelligent
  Transportation Systems (ITSC)}.\hskip 1em plus 0.5em minus 0.4em\relax IEEE,
  2017, pp. 353--359.

\bibitem{mohamed2020social}
A.~Mohamed, K.~Qian, M.~Elhoseiny, and C.~Claudel, ``Social-stgcnn: A social
  spatio-temporal graph convolutional neural network for human trajectory
  prediction,'' in \emph{Proceedings of the IEEE/CVF Conference on Computer
  Vision and Pattern Recognition (CVPR)}, 2020, pp. 14\,424--14\,432.

\bibitem{schulz2018multiple}
J.~Schulz, C.~Hubmann, J.~L{\"o}chner, and D.~Burschka, ``Multiple model
  unscented {Kalman} filtering in dynamic {Bayesian} networks for intention
  estimation and trajectory prediction,'' in \emph{2018 IEEE 21st International
  Conference on Intelligent Transportation Systems (ITSC)}.\hskip 1em plus
  0.5em minus 0.4em\relax IEEE, 2018, pp. 1467--1474.

\bibitem{li2020unpredictable}
J.~Li, J.~He, Y.~Li, and X.~Guan, ``Unpredictable trajectory design for mobile
  agents,'' in \emph{2020 American Control Conference (ACC)}.\hskip 1em plus
  0.5em minus 0.4em\relax IEEE, 2020, pp. 1471--1476.

\bibitem{ljung1998system}
L.~Ljung, ``System identification,'' in \emph{Signal Analysis and
  Prediction}.\hskip 1em plus 0.5em minus 0.4em\relax Springer, 1998, pp.
  163--173.

\bibitem{welch1995introduction}
G.~Welch, G.~Bishop \emph{et~al.}, ``An introduction to the {Kalman} filter,''
  1995.

\bibitem{gillijns2007unbiased}
S.~Gillijns and B.~De~Moor, ``Unbiased minimum-variance input and state
  estimation for linear discrete-time systems,'' \emph{Automatica}, vol.~43,
  no.~1, pp. 111--116, 2007.

\bibitem{molloy2016discrete}
T.~L. Molloy, D.~Tsai, J.~J. Ford, and T.~Perez, ``Discrete-time inverse
  optimal control with partial-state information: A soft-optimality approach
  with constrained state estimation,'' in \emph{2016 IEEE 55th Conference on
  Decision and Control (CDC)}.\hskip 1em plus 0.5em minus 0.4em\relax IEEE,
  2016, pp. 1926--1932.

\bibitem{li2023topology}
Y.~Li, J.~He, C.~Chen, and X.~Guan, ``Topology inference for network systems:
  Causality perspective and non-asymptotic performance,'' \emph{IEEE
  Transactions on Automatic Control}, to be published, 2023.

\bibitem{bertsekas2012dynamic}
D.~Bertsekas, \emph{Dynamic programming and optimal control: Volume I}.\hskip
  1em plus 0.5em minus 0.4em\relax Athena Scientific, 2012, vol.~1.

\bibitem{richter1995estimating}
P.~H. Richter, ``Estimating errors in least-squares fitting,'' \emph{The
  Telecommunications and Data Acquisition Report}, 1995.

\bibitem{stellato2020osqp}
B.~Stellato, G.~Banjac, P.~Goulart, A.~Bemporad, and S.~Boyd, ``Osqp: An
  operator splitting solver for quadratic programs,'' \emph{Mathematical
  Programming Computation}, vol.~12, no.~4, pp. 637--672, 2020.

\bibitem{kwakernaak1972maximally}
H.~Kwakernaak and R.~Sivan, ``The maximally achievable accuracy of linear
  optimal regulators and linear optimal filters,'' \emph{IEEE Transactions on
  Automatic Control}, vol.~17, no.~1, pp. 79--86, 1972.

\bibitem{bellman1966dynamic}
R.~Bellman, ``Dynamic programming,'' \emph{Science}, vol. 153, 1966.

\bibitem{bitmead1991riccati}
R.~R. Bitmead and M.~Gevers, ``Riccati difference and differential equations:
  Convergence, monotonicity and stability,'' \emph{The Riccati Equation}, pp.
  263--291, 1991.

\bibitem{lofberg2004yalmip}
J.~Lofberg, ``Yalmip: A toolbox for modeling and optimization in matlab,'' in
  \emph{2004 IEEE International Conference on Robotics and Automation
  (ICRA)}.\hskip 1em plus 0.5em minus 0.4em\relax IEEE, 2004, pp. 284--289.

\bibitem{sturm1999using}
J.~F. Sturm, ``Using sedumi 1.02, a matlab toolbox for optimization over
  symmetric cones,'' \emph{Optimization Methods and Software}, vol.~11, no.
  1-4, pp. 625--653, 1999.

\bibitem{ding2021robopheus}
X.~Ding, H.~Wang, H.~Li, H.~Jiang, and J.~He, ``Robopheus: A virtual-physical
  interactive mobile robotic testbed,'' 2021.

\bibitem{cai2017convergent}
X.~Cai, Y.~Ding, and S.~Li, ``Convergent properties of riccati equation with
  application to stability analysis of state estimation,'' \emph{Mathematical
  Problems in Engineering}, vol. 2017, 2017.

\end{thebibliography}

\begin{IEEEbiographynophoto}{Chendi Qu} received the B.E. degree in the Department of Automation from Tsinghua University, Beijing, China, in 2021. 
She is currently working toward the Ph.D. degree with the Department of Automation, Shanghai Jiao Tong University, Shanghai, China. 
She is a member of Intelligent Wireless Networks and Cooperative Control group. 
Her research interests include robotics, security of cyber-physical system, and distributed optimization and learning in multi-agent networks. 
\end{IEEEbiographynophoto}

\begin{IEEEbiographynophoto}{Jianping He} 
(SM’19) is an Associate Professor in the Department of
Automation at Shanghai Jiao Tong University. He received the Ph.D. degree in control science and engineering from Zhejiang University, Hangzhou, China, in 2013, and had been a research fellow in the Department of Electrical and Computer Engineering at University of Victoria, Canada, from Dec. 2013 to Mar. 2017. His research interests mainly include the distributed learning,
control and optimization, security and privacy in network systems.

Dr. He serves as an Associate Editor for IEEE Trans. Control of Network Systems, IEEE Open Journal of Vehicular Technology, and KSII Trans. Internet and Information Systems. He was also a Guest Editor of IEEE TAC, IEEE TII, International Journal of Robust and Nonlinear Control, etc. He was the winner of Outstanding Thesis Award, Chinese Association of Automation, 2015. He received the best paper award from IEEE WCSP'17, the best conference paper
award from IEEE PESGM'17, and was a finalist for the best student paper award from IEEE ICCA'17, and the finalist best conference paper award from IEEE VTC'20-FALL.
\end{IEEEbiographynophoto}

\begin{IEEEbiographynophoto}{Xiaoming Duan} 
is an assistant professor in the Department of Automation at Shanghai Jiao Tong University.
He obtained his B.E. degree in Automation from the
Beijing Institute of Technology in 2013, his Master’s Degree in Control Science and Engineering from Zhejiang
University in 2016, and his Ph.D. degree in Mechanical
Engineering from the University of California at Santa
Barbara in 2020. He was a postdoctoral fellow in
the Oden Institute for Computational Engineering and
Sciences at the University of Texas at Austin in 2021.
His research interests include robotics, multi-agent
systems, and autonomous systems.
\end{IEEEbiographynophoto}

\end{document}